\documentclass[12pt]{article}

\usepackage{hyperref}
\hypersetup{colorlinks=true,linkcolor=red,citecolor=green,filecolor=magneta,urlcolor=cyan}

\usepackage{amsfonts,amssymb,amsmath,amsthm}
\usepackage[a4paper,hmargin={2.5cm,2.5cm},vmargin={2.3cm,2.3cm}]{geometry}

\usepackage{graphicx}

\DeclareMathOperator{\real}{Re}
\def\Rset{{\mathbb{R}}}
\def\Cset{{\mathbb{C}}}

\def\Nset{{\mathbb{N}}}
\def\eu{\ensuremath{\mathrm{e}}}
\def\iu{\ensuremath{\mathrm{i}}}
\def\du{\ensuremath{\mathrm{d}}}

\newtheorem{theorem}{Theorem}
\newtheorem{remark}[theorem]{Remark}
\newtheorem{proposition}[theorem]{Proposition}

\newtheorem{corollary}[theorem]{Corollary}
\theoremstyle{definition}

\newenvironment{pf}{\medskip\noindent{\bf Proof.}\enspace}   
{\hfill\newline\smallskip}

\begin{document}

%\begin{frontmatter}

\title{The Prabhakar or three parameter Mittag--Leffler function: theory and application \footnote{This is an updated preprint of the paper published in \emph{Communications in Nonlinear Science and Numerical Simulation}, \url{http://dx.doi.org/10.1016/j.cnsns.2017.08.018} Volume 56, 2018, Pages 314-329, doi: 10.1016/j.cnsns.2017.08.018 - Some misprints in equations (\ref{HE1a}) and (\ref{HE1}) have been corrected.}}

\author{Roberto Garra \\
\small Universit\`a La Sapienza, Dipartimento di Scienze Statistiche, Roma, Italy \\
\small \texttt{roberto.garra@sbai.uniroma1.it}
\and
Roberto Garrappa \\
\small Universit\`a degli Studi di Bari, Dipartimento di Matematica, Bari, Italy\\ 
\small \texttt{roberto.garrappa@uniba.it}
}

\date{May 12$^{\text{th}}$, 2017} 

\maketitle

\begin{abstract}
The Prabhakar function (namely, a three parameter Mittag-Leffler function) is investigated. This function plays a fundamental role in the description of the anomalous dielectric properties in disordered materials and heterogeneous systems manifesting simultaneous nonlocality and nonlinearity and, more generally, in models of Havriliak-Negami type. After reviewing some of the main properties of the function, the asymptotic expansion for large arguments is investigated in the whole complex plane and, with major emphasis, along the negative semi-axis. Fractional integral and derivative operators of Prabhakar type are hence considered and some nonlinear heat conduction equations with memory involving Prabhakar  derivatives are studied.

\end{abstract}

%% main text
\section{Introduction}\label{S:Introduction}

%\begin{itemize}

%	\item in \cite{Pandey2017} sono studiati i modelli di rilassamento  dieltterico con la Prabhakar (vicino ad HN)
%	\item \cite{DerakhshanAhmadiMohammadrezaAnsariKhoshsiar2016} si studia la stabilita
%	\item in \cite{LiemertSandevKantz2017} it is shown that that the  Prabhakar integral operator appears in the analysis of correlation functions of  a Generalized Langevin equation with tempered memory kernel 
%  \item in \cite{EshaghiAnsari2016}  we consider a fractional boundary value problem including the Prabhakar fractional derivative. We obtain associated Green function for this fractional boundary value problem and get a Lyapunov-type inequality for it.
%  \item in \cite{AskariAnsari2016} we introduce a generalization of the Hilfer-Prabhakar derivative and obtain the Euler-Lagrange equations and  Hamiltonian formulation with respect to this fractional derivative in the theory of fractional calculus of variations. Also, we get a sufficient condition for optimality.
%\item In \cite{PanevaKonovska2013,PanevaKonovska2014} si studia la convergenza ... we study the convergence in series of such function
% \item in \cite{FigueiredoCamargoCapelasOliveiraVaz2012} is used to find a solution of a Cauchy problem for the fractional telegraph equation
%\end{itemize}

The Mittag--Leffer (ML) function was introduced, at the beginning of the 19-th century, by the Swedish mathematician Magnus Gustaf Mittag--Leffler in connection with methods for summation of divergent series \cite{Mittag-Leffler1902,MittagLeffler1904}. Successively, several other functions of this type were proposed and investigated (see, for instance, \cite{AlBassamLuchko1995,GarraPolito2013,Gerhold2012,GorenfloKilbasRogosin1998,KilbasSaigo1995,KilbasKorolevaRogosin2013,Kiryakova2010} and the recent monograph \cite{GorenfloKilbasMainardiRogosin2014} completely devoted to ML functions). 

%Because of the possibility of defining functions with a different number of parameters and due to the dependence of these functions on such parameters, it is often referred to ML functions as a family of functions. 

The main interest in studying ML functions is related to their importance in fractional calculus where they play the same fundamental role of the exponential function in integer order calculus.

In recent years a three parameter ML function, also known as the Prabhakar\footnotemark\, function \cite{Prabhakar1971}, is attracting a remarkable attention. For any $z \in \Cset$, this function is defined as
\begin{equation}\label{eq:Prabhakar}
	E_{\alpha,\beta}^{\gamma}(z) = \frac{1}{\Gamma(\gamma)}\sum_{k=0}^{\infty} \frac{ \Gamma(\gamma+k) z^{k}}{k! \Gamma(\alpha k + \beta)}
	, \quad 
	\alpha, \beta, \gamma \in \Cset
	, \quad \real(\alpha) > 0 ,
\end{equation}
where, as usual, $\Gamma(\cdot)$ denotes the Euler's gamma function. $E_{\alpha,\beta}^{\gamma}(z)$ is an entire function of order $\rho = 1/\real(\alpha)$ and type $\sigma=1$ \cite{HauboldMathaiSaxena2011,Lavault2017}. 

\footnotetext{The Indian mathematician Tilak Raj Pabhakar, who this function is named after, completed the M.A. and B.A. at the Benaras Hindu University (India) and obtained his doctorate form the University of Delhi (India) in 1970 with a thesis on ``Integral Equations and Special Functions''. He taught in India at the S.D. College in Muzzaffarnagar, at the Multani Mal Modi College in Modinagar (where he became Head of the Department of Mathematics) and at the Ramjas College of the University of Delhi. He was professor of Mathematics at the National Institute of Electricity and Electronics in Boomerdes (Algeria) when he died in 1982 \cite{Ganita1982}.}

Some efforts have been recently undertaken to study the main properties of the Prabhakar function (e.g., see \cite{KilbasSaigoSaxena2004,NigmatullinKhamzinBaleanu2016,SaxenaSaigo2005}) with major emphasis on the  complete monotonicity \cite{CapelasMainardiVaz2011,TomovskiPoganySrivastava2014,MainardiGarrappa2015_JCP}.

However, as pointed out in \cite[Page 101]{GorenfloKilbasMainardiRogosin2014}, the asymptotic behavior of the Prabhakar function ``has not yet been described in an explicit form'' in the whole complex plane. The reason for this lack is mainly due to the dependence on the 3 parameters $\alpha$, $\beta$ and $\gamma$ which makes quite difficult the explicit derivation of the coefficients in the asymptotic expansion. 

Formula (\ref{eq:Prabhakar}) provides a representation of the Prabhakar function in terms of a Taylor series centered at the origin and it is therefore especially useful to describe its behavior as $|z|\to 0$. 

The derivation of a general expansion for large arguments, i.e. as $|z|\to \infty$, still remains an under-explored subject. In some limited cases the problem can be faced by means of the Laplace transform thanks to which it is possible to show (e.g., see \cite{MainardiGarrappa2015_JCP}) that for $t$ real and positive it is 
\[
	E_{\alpha,\beta}^{\gamma}(-t^{\alpha}) \sim  \sum_{k=0}^{\infty}\binom{-\gamma}{k} \frac{t^{-\alpha (\gamma+k)}}{\Gamma\bigl(\beta-\alpha(\gamma+k)\bigr)} , \quad t \to \infty
\]
thus providing an asymptotic expansion on the negative real axis (note that the first term in the above summation cancels out when $\beta=\alpha\gamma$).

Anyway, as also noticed in \cite{GorenfloKilbasMainardiRogosin2014}, asymptotic expansions of the Prabhakar function can be found from its representation as a generalized Fox-Wright function \cite{Fox1928,Wright1935}. Several authors have indeed studied the asymptotic behavior of the Fox-Wright functions and these results are surely useful for the Prabhakar function too.

One of the main aims of this paper is to apply to the Prabhakar function some of the results on the asymptotic expansion of Fox-Wright functions (in particular those published in 2010 by Paris \cite{Paris2010}) and explicitly represent, in the whole complex plane, the asymptotic expansions for large arguments; at the same time we show an algorithmic procedure for the computation of the coefficients in the expansion and we perform a deeper examination of the behavior of the function on the negative semi axis (a subject which is of prominent importance for applications).

Most of the interest in the Prabhakar function is related to the description of relaxation and response in anomalous dielectrics of Havriliak--Negami type (e.g., see \cite{GarrappaMainardiMaione2016,StanislavskyWeron2016_FCAA,Pandey2017}), a model of complex susceptibility introduced to keep into account the simultaneous nonlocality and nonlinearity observed in the response of disordered materials and heterogeneous systems \cite{Miskinis2009}. Further applications of the Prabhakar function are however encountered in probability theory \cite{GorskaHorzelaBaratekPensonDattoli2016,James2010,PoganyTomovski2016}, in the study of stochastic processes \cite{DOvidioPolito2013,PolitoScalas2016} and of systems with strong anisotropy \cite{ChamatiTonchev2006,Tonchev2007}, in fractional viscoelasticity \cite{GiustiColombaro2017}, in the solution of some fractional boundary-value problems \cite{BazhlekovaDimovski2013,BazhlekovaDimovski2014,EshaghiAnsari2016,FigueiredoCamargoCapelasOliveiraVaz2012,LuchkoSrivastava1995}, in the description of dynamical models of spherical stellar systems \cite{AnVanHeseBaes2012} and in connection with other fractional or integral differential equations \cite{AskariAnsari2016,LiemertSandevKantz2017,KilbasSaigoSaxena2002}. 

The time-evolution of polarization processes in Havriliak-Negami models and of the other above listed phenomena can be suitably described by integral and differential operators based on the Prabhakar function \cite{ChaurasiaPandey2010,DOvidioPolito2013,GarraGorenfloPolitoTomovski2014,Garrappa2016_CNSNS,PolitoTomovski2016,GuptaShaktawatKumar2016} (see also \cite{DerakhshanAhmadiMohammadrezaAnsariKhoshsiar2016} for an analysis of stability properties). In this paper, after providing an outline of Prabhakar operators (whose introduction is very recent in literature), we study some nonlinear heat conduction equations with memory involving Prabhakar derivatives; we find the exact solution and show the asymptotic behavior as $t \to \infty$.

This paper is organized as follows. We first review, in Section \ref{S:Property}, the main properties of the Prabhakar function and we present integral and derivative operators of Prabhakar type. In Section \ref{S:AsymptoticExpansion} we illustrate the results on the asymptotic expansions of the function for large arguments and we discuss the behavior on the real negative semi-axis (the most technical details concerning the algorithm for the derivation of the coefficients are collected in the Appendix at the end of paper). Section \ref{S:Heat} is devoted to the study of two nonlinear heat conduction equations with derivative of Prabhakar type and, finally, some concluding remarks are presented in Section \ref{S:Concluding}.

\section{Properties of the Prabhakar function}\label{S:Property}

For completeness of exposition, we summarize some of the most remarkable properties of the Prabhakar function. For a more in-depth discussion we refer, for instance, to \cite{Prabhakar1971,HauboldMathaiSaxena2011,KilbasSaigoSaxena2004,KurulayBayram2012,MathaiSaxenaHaubold2010,PanevaKonovska2013,PanevaKonovska2014}.

\subsubsection*{Integer values of $\gamma$}

Since for $k\ge 1$ it is $\Gamma(\gamma+k)/\Gamma(\gamma) = \gamma(\gamma+1)\cdots(\gamma+k-1)$, when $\gamma=0$ we clearly have $E_{\alpha,\beta}^{0}(z) = 1/\Gamma(\beta)$. When $\gamma=1$, $E_{\alpha,\beta}^{1}(z)=E_{\alpha,\beta}(z)$ is instead the standard ML function with two parameters. Thus, thanks to a formula for the reduction in the third parameter \cite{Prabhakar1971} 
\begin{equation}\label{eq:ReductionGamma}
	E_{\alpha,\beta}^{\gamma+1}(z) = 
	\frac{E_{\alpha,\beta-1}^{\gamma}(z) + (1-\beta+\alpha\gamma) E_{\alpha,\beta}^{\gamma}(z)}{\alpha \gamma} ,
\end{equation}
it is possible to write the Prabhakar function, when $\gamma =k \in \Nset$, in terms of $k$ values of the standard ML function; for instance, we can see that
\[
	\begin{aligned}
		E_{\alpha,\beta}^{2}(z) &= \frac{1}{\alpha} \left[ E_{\alpha,\beta-1}(z) + (1-\beta+\alpha) E_{\alpha,\beta}(z) \right] \\
		E_{\alpha,\beta}^{3}(z) &= \frac{1}{2\alpha^2} \left[ E_{\alpha,\beta-2}(z) + (1-\alpha) E_{\alpha,\beta-1}(z) + (1-\beta+2\alpha)(1-\beta+\alpha) E_{\alpha,\beta}(z) \right] \\
	\end{aligned}
\]

Moreover, whenever $\gamma=-j$, with $j \in \Nset$, a simple computation allows to verify that the Prabhakar function is the $j$-degree polynomial
\[
	E_{\alpha,\beta}^{-j}(z) = \sum_{k=0}^{j} (-1)^k\binom{j}{k} \frac{z^{k}}{\Gamma(\alpha k + \beta)} .
\]

\subsubsection*{The Laplace transform}

As in the case of the standard Mittag-Leffler function, a direct Laplace transform (LT) of the Prabhakar function $E_{\alpha,\beta}^{\gamma}(z)$ is not known but it is possible to explicitly represent the LT of a generalization of this function since
\begin{equation}\label{eq:ML3_LT}
	{\mathcal L} \left( t^{\beta-1} E^{\gamma}_{\alpha,\beta}( t^{\alpha} z)   ; s \right) = \frac{s^{\alpha\gamma-\beta}}{(s^{\alpha} - z)^{\gamma}}
	, \quad \Re(s)>0 \, \text{ and } \, |s|^{\alpha} > |z|  .
\end{equation}

The very simple analytical formulation of the LT, especially when compared to the series representation (\ref{eq:Prabhakar}), makes the LT a reliable tool for the practical computation of the Prabhakar function after numerical inversion of (\ref{eq:ML3_LT}), as already discussed in \cite{Garrappa2015_SIAM}.

\subsubsection*{Integrals and derivatives}

Formulas for integration and derivation of the Prabhakar function can be easily obtained by means of the LT or thanks to a term-by-term integration or differentiation \cite{Prabhakar1971,KilbasSaigoSaxena2004,SaxenaSaigo2005}. In particular, we recall here that for any $t \in \Rset_{+}$ and $\lambda \in \Cset$ it is  
\[
	\int_{0}^{t} u^{\beta-1} E_{\alpha,\beta}^{\gamma}(\lambda u^{\alpha}) \du u = t^{\beta} E_{\alpha,\beta+1}^{\gamma}(\lambda t^{\alpha})
\]
and for any $m \in \Nset$
\[
	\frac{\du^{m}}{\du t^{m}} \left[ t^{\beta-1} E_{\alpha,\beta}^{\gamma}(\lambda t^{\alpha}) \right] = t^{\beta-m-1} E_{\alpha,\beta-m}^{\gamma} (\lambda t^{\alpha}) ;
\]
moreover, it is also possible to see that
\[
	\frac{\du }{\du z} E_{\alpha,\beta}^{\gamma}(z) = \gamma E_{\alpha,\alpha+\beta}^{\gamma+1}(z)
\]
and this formula can be easily iterated to obtain 
\begin{equation}\label{eq:DerivPrabhakar1}
	\frac{\du^m }{\du z^m} E_{\alpha,\beta}^{\gamma}(z) = \gamma(\gamma+1) \cdots (\gamma+m-1) E_{\alpha,m\alpha+\beta}^{\gamma+m}(z) .
\end{equation}

The Dzhrbashyan's formula \cite{Djrbashian1993} for the first derivative of the Mittag-Leffler function can be extended to the Prabhakar  function as stated in the following result.

\begin{proposition}
Let $z\not=0$. Then it is
\begin{equation}\label{eq:DerivPrabhakar2}
	\frac{\du }{\du z} E_{\alpha,\beta}^{\gamma}(z) = \frac{E_{\alpha,\beta-1}^{\gamma}(z) + (1-\beta)E_{\alpha,\beta}^{\gamma}(z)}{\alpha z} .
\end{equation}
\end{proposition}
\begin{proof}
By a term-by-term derivation it is immediate to see that
\[
	\frac{\du }{\du z} E_{\alpha,\beta}^{\gamma}(z) 
	= \frac{1}{\Gamma(\gamma)} \sum_{k=1}^{\infty} \frac{\Gamma(\gamma+k) z^{k-1}}{(k-1)! \Gamma(\alpha k +\beta)} .
\]

Moreover
\[
	E_{\alpha,\beta-1}^{\gamma}(z) + (1-\beta)E_{\alpha,\beta}^{\gamma}(z)
	= \frac{1}{\Gamma(\gamma)} \sum_{k=0}^{\infty} \frac{\Gamma(\gamma+1)}{k!} \left[ \frac{z^{k}}{\Gamma(\alpha k + \beta-1)} + \frac{(1-\beta)z^{k}}{\Gamma(\alpha k + \beta)} \right]
\]
and, after replacing 
\[
	\frac{1}{\Gamma(\alpha k + \beta-1)} = \frac{(\alpha k + \beta-1)}{\Gamma(\alpha k+\beta)}
\]
the proof immediately follows. 
\end{proof}

The comparison of formulas (\ref{eq:DerivPrabhakar1}) and (\ref{eq:DerivPrabhakar2}) allows us to derive a further formula, in addition to (\ref{eq:ReductionGamma}), for the reduction in the third parameter $\gamma$ according to the following result whose proof is immediate and hence omitted.
\begin{corollary}
Let $z\not=0$. Than it is
\[
	E_{\alpha,\beta}^{\gamma+1}(z) = 
	\frac{E_{\alpha,\beta-\alpha-1}^{\gamma}(z) + (1-\beta+\alpha) E_{\alpha,\beta-\alpha}^{\gamma}(z)}{\alpha \gamma z} .
\]
\end{corollary}

Finally we see that the Riemann-Liouville integral ${}_{0}J^{\rho}_{t}$ and the Riemann-Liouville derivative ${}_{0}D^{\rho}_{t}$ of order $\rho>0$ of the Prabhakar function are respectively given by
\[
	{}_{0}J^{\rho}_{t} \left[ t^{\beta-1} E_{\alpha,\beta}^{\gamma}(\lambda t^{\alpha}) \right] = 
	\left[ t^{\beta+\rho-1} E_{\alpha,\beta+\rho}^{\gamma}(\lambda t^{\alpha}) \right]
\]
and
\[
	{}_{0}D^{\rho}_{t} \left[ t^{\beta-1} E_{\alpha,\beta}^{\gamma}(\lambda t^{\alpha}) \right] = 
	\left[ t^{\beta-\rho-1} E_{\alpha,\beta-\rho}^{\gamma}(\lambda t^{\alpha}) \right] .
\]

\subsubsection*{Complete monotonicity}

We recall that for a function $f:(0,+\infty) \to \Rset$, the complete monotonicity (CM) states that $f$ has derivatives of all orders on $(0,+\infty)$ and $(-1)^{k}f^{(k)}(t) \ge 0$ for any $k\in\Nset$ and $t > 0$.

This property is of particular importance for the physical acceptability of related models since, as discussed for instance by Hanyga \cite{Hanyga2005b}, the CM ensures the monotone decay of the energy in isolated systems. For this reason, and in view of the use in dielectric models, the CM of the Prabhakar function has been studied in several works \cite{CapelasMainardiVaz2011, MainardiGarrappa2015_JCP,TomovskiPoganySrivastava2014} and nowadays the most general result states that the Prabhakar function is CM for the following parameters
\[
	0 < \alpha \le 1 , \quad 0 < \alpha \gamma \le \beta \le 1.
\]

\subsubsection*{Relationship with Fox-Wright functions}

It is immediate to verify that the Prabhakar function is a special instance of the Fox--Wright functions \cite{Fox1928,Wright1935}, a very general class of multi-parameter functions defined by
\begin{equation}\label{eq:FoxWright}
	_{p}\Psi_{q}(z) \equiv \, 
	_{p}\Psi_{q} \left[ \begin{array}{c} (\rho_1,a_1),\dots,(\rho_p,a_p) \\ (\sigma_1,b_1),\dots,(\sigma_q,b_q) \end{array} ; z \right] = \sum_{k=0}^{\infty} \frac{z^{k}}{k!} \frac
		{\displaystyle\prod_{r=1}^{p} \Gamma(\rho_{r}k + a_{r}) } 
		{\displaystyle\prod_{r=1}^{q} \Gamma(\sigma_{r} k + b_{r}) }
	,
\end{equation}
where $p$ and $q$ are integers and $\rho_r, a_r, \sigma_r,b_r$ are real or complex parameters. Indeed, we can easily verify that $E_{\alpha,\beta}^{\gamma}(z)$ is proportional to $_{1}\Psi_{1}(z)$ since
\begin{equation}\label{eq:PrabhakarFoxWright}
	E_{\alpha,\beta}^{\gamma}(z) 
	= \frac{1}{\Gamma(\gamma)} \, _{1}\Psi_{1} \left[ \begin{array}{c} (1,\gamma) \\ (\alpha,\beta) \end{array} ; z \right] .
\end{equation}

\subsubsection*{Fractional integral and derivative involving the Prabhakar function}\label{S:Operators}

In 1967 Havriliak and Negami \cite{HavriliakNegami1967} formulated a new model of complex relaxation in order to take into account asymmetry
and broadness in the shape of the permittivity spectrum of some polymers. The complex susceptibility in the Havriliak-Negami model is characterized by one relaxation time $\tau_{\star}$ and two real powers $\alpha$ and $\gamma$ 
\begin{equation}\label{eq:HN}
	\hat{\chi}_{\text{\tiny{HN}}}(\iu \omega) = \frac{1}{\left( 1 + \bigl( \iu \tau_{\star} \omega \bigr)^{\alpha} \right)^{\gamma}} .
\end{equation}

The corresponding equations for the time-domain evolution of response and relaxation functions are hence expressed in terms of the operators \cite{GarrappaMainardiMaione2016} 
\[
	\bigl({}_{0}J^{\alpha}_{t} + \lambda \bigr)^{\gamma}
	, \quad \quad
	\bigl({}_{0}D^{\alpha}_{t} + \lambda \bigr)^{\gamma}
	, 
\]
where $\lambda = \tau_{\star}^{-\alpha}$ and ${}_{0}J^{\alpha}_{t}$ and ${}_{0}D^{\alpha}_{t}$ denote, respectively, the fractional integral and the fractional derivative of Riemann-Liouville type \cite{Diethelm2010}. In the past several efforts have been done to properly characterize the above operators (see, for instance, \cite{NigmatullinRyabov1997,NovikovWojciechowskiKomkovaThiel2005}). One of the most powerful is however the approach based on the generalized Prabhakar function 
\[
	e_{\alpha,\beta}^{\gamma}(t;\lambda) = t^{\beta-1} E_{\alpha,\beta}^{\gamma}(t^{\alpha}\lambda)
\]

The corresponding operators can be indeed represented by means of the convolution integrals with $e_{\alpha,\beta}^{\gamma}$ in the kernel \cite{PolitoTomovski2016}
\[
	\bigl({}_{0}J^{\alpha}_{t} + \lambda\bigr)^{\gamma} f(t) 
	\equiv e_{\alpha,\alpha\gamma}^{\gamma} (t;-\lambda) \star f(t) 		
	= \int_{0}^{t} e_{\alpha,\alpha\gamma}^{\gamma} (t-u;-\lambda)  f(u) \, \du u,
\]
and
\[
	\bigl({}_{0}D^{\alpha}_{t} + \lambda\bigr)^{\gamma} f(t)
	\equiv \frac{\du^m}{\du t^m} \left( e_{\alpha,\alpha\gamma}^{\gamma} (t;-\lambda) \star f(t) 	\right)
	= \frac{\du^m}{\du t^m} \int_{0}^{t} e_{\alpha,m-\alpha\gamma}^{-\gamma} (t-u;-\lambda) f(u) \, \du u , 
\]
with $m = \left\lceil \alpha \gamma \right\rceil$; although we keep the presentation of the above operators at the most general level, in practical applications $\alpha \gamma$ is confined to the interval $(0,1)$ and hence one should just consider $m=1$. 

A regularization in the Caputo sense of the Prabhakar derivative has been recently proposed in \cite{GarraGorenfloPolitoTomovski2014}
\begin{equation}\label{CP}
	{}^{\text{\tiny{C}}}{}{\bigl({}_{0}D^{\alpha}_{t} + \lambda \bigr)^{\gamma}}  f(t)
	\equiv e_{\alpha,\alpha\gamma}^{\gamma} (t;-\lambda) \star \frac{\du^m}{\du t^m} f(t) 
	= \int_{0}^{t} e_{\alpha,m-\alpha\gamma}^{-\gamma} (t-u;-\lambda) f^{(m)}(u) \, \du u ,
\end{equation}
where we followed the notation later introduced in \cite{Garrappa2016_CNSNS,GarrappaMainardiMaione2016}. 

Between the two derivatives $\bigl({}_{0} D^{\alpha}_{t} + \lambda\bigr)^{\gamma}$ and ${}^{\text{\tiny{C}}}{}{\bigl({}_{0}D^{\alpha}_{t} + \lambda \bigr)^{\gamma}}$ there exist a relationship, similar to the one existing between the Riemann-Liouville and the Caputo fractional derivatives, given by
\begin{equation}\label{eq:HN_RL_Caputo_Relat}
	{}^{\text{\tiny{C}}}{}{\bigl({}_{0}D^{\alpha}_{t} + \lambda \bigr)^{\gamma}} f(t)
	= \bigl({}_{0}D^{\alpha}_{t} + \lambda\bigr)^{\gamma} \left( f(t) - \sum_{j=0}^{m-1} \frac{t}{j!} f(0^+) \right) .
\end{equation}

A characterization in terms of Gr\"{u}nwald-{L}etnikov differences of the above operators has been successively studied in \cite{Garrappa2016_CNSNS}. 

We refer the reader to \cite{Prabhakar1971,GarraGorenfloPolitoTomovski2014}, and in particular to the appendix in \cite{GarrappaMainardiMaione2016}, for a more detailed discussion concerning these operators, their applications and properties.

\section{Asymptotic expansion for large arguments}\label{S:AsymptoticExpansion}

To derive asymptotic expansions for large arguments in the whole complex plane it is useful to exploit, thanks to the relationship (\ref{eq:PrabhakarFoxWright}), some of the results derived for the Fox--Wright functions \cite{Fox1928,Wright1935}.

Asymptotic expansions of the Fox--Wright functions were mainly investigated by Edward Maitland Wright who first provided in \cite{Wright1935} an account of the main results and successively published the corresponding proofs in \cite{Wright1940} (other results are also discussed in \cite{Wrigth1940_PTRSL}). Some other related results were successively presented by Braaksma in the long paper \cite{Braaksma1964} mainly dedicated to $H$-functions.

On the basis of these works, the asymptotic expansion of the Fox--Wright functions for large arguments has been studied by Paris \cite{Paris2010} in a paper describing an efficient algorithm for the derivation of the coefficients in the asymptotic expansion of Fox--Wright functions. 

Their use for the Prabhakar function is not however immediate since the coefficients are not provided in an explicit form because of their dependence on multiple parameters.

We therefore review and summarize the work in \cite{Paris2010} and we provide the explicit representation for the asymptotic expansions of the Prabhakar function; we hence give the first few corresponding coefficients, although the description of the algorithm for their derivation, which involves more technical issues, is confined to the Appendix at the end of the paper. We also propose a more in-depth and original discussion of the behavior on the negative semi-axis.

\subsection{Expansion in the whole complex plane} \label{S2} 

Let us assume $\alpha$ be real positive and $\beta$ and $\gamma$ arbitrary complex numbers with $\gamma\not=-1,-2,\dots$. By following the approach proposed in \cite{Paris2010}, we introduce the functions
\[
	H(z) = \frac{z^{-\gamma}}{\Gamma(\gamma)} \sum_{k=0}^{\infty} \frac{(-1)^{k} \Gamma(k+\gamma)}{k! \Gamma(\beta-\alpha(k+\gamma))} z^{-k}
\]
and
\[
	F(z) = \frac{1}{\Gamma(\gamma)} \eu^{z^{1/\alpha}} z^{\frac{\gamma-\beta}{\alpha}} \frac{1}{\alpha^{\gamma}} \sum_{k=0}^{\infty} c_{k} z^{-\frac{k}{\alpha}} ,
\]
with the coefficients $c_{k}$ of $F(z)$ (which depend on the parameters $\alpha$, $\beta$ and $\gamma$) obtained as the coefficients in the inverse factorial expansion of 
\begin{equation}\label{eq:InverseFactorialExpansion}
	F_{\alpha,\beta}^{\gamma}(s) := \frac{\Gamma(\gamma+s)\Gamma(\alpha s + \psi)}{\Gamma(s+1)\Gamma(\alpha s + \beta)}
	= \alpha^{1-\gamma} \left( 1 + \sum_{j=1}^{\infty} \frac{c_{j}}{(\alpha s + \psi )_j} \right)
\end{equation}
for $|s|\to \infty$ in $|\arg(s)| \le \pi - \epsilon$ and any arbitrarily small $\epsilon >0$; as usual, $(x)_j = x(x+1)\cdots(x+j-1)$ denotes the Pochhammer symbol and, just for shortness, we put
\[
	\psi = 1-\gamma+\beta .
\]

The main results for the asymptotic expansions of the Prabhakar function can be given thanks to the following theorems (for the proofs see \cite{Wright1935}, \cite[Section 2.3]{ParisKaminski2001} or \cite[Theorem 1, 2 and 3]{Paris2010}).

\begin{theorem}\label{thm:Expansion1}
Let $0< \alpha <2$. Then
\[
	E_{\alpha,\beta}^{\gamma}(z) \sim
	\left\{\begin{array}{ll}
		F(z) + H(ze^{\mp \pi i}) \quad & \textrm{if } |\arg(z)| \le \frac{1}{2} \pi \alpha \\
		H(ze^{\mp \pi i}) & \textrm{if } |\arg(-z)| < \frac{1}{2} \pi (2-\alpha)  \
	\end{array} \right.
\]
as $|z|\to \infty$ with the sign in $H(ze^{\mp \pi i})$ chosen according as $z$ lies  in the upper or lower half-plane respectively.
\end{theorem}

\begin{theorem}\label{thm:Expansion2}
Let $\alpha =2$ and $|\arg(z)| \le \pi$. Then
\[
	E_{\alpha,\beta}^{\gamma}(z) \sim F(z) + F(z\eu^{\mp 2\pi\iu}) + H(z\eu^{\mp \pi\iu})
	, \quad \text{as } |z|\to \infty,
\]
with the sign in $F(z\eu^{\mp 2\pi\iu})$ and $H(ze^{\mp \pi i})$ chosen according as $z$ lies in the upper or lower half-plane respectively.
\end{theorem}

\begin{theorem}\label{thm:Expansion3}
Let $\alpha >2$ and $|\arg(z)| \le \pi$. Then
\[
	E_{\alpha,\beta}^{\gamma}(z) \sim \sum_{r=-P}^{P} F(z\eu^{2\pi\iu r}) 
	, \quad \text{as } |z|\to \infty,
\]
with $P$ the integer such that $2P+1$ is the smallest odd integer satisfying $2P+1>\frac{1}{2} \alpha$.
\end{theorem}

The evaluation of the coefficients $c_k$ in $F(z)$ is not an easy task and, indeed, each coefficient is function (of increasing complexity as $k$ increases) of the three parameters $\alpha$, $\beta$ and $\gamma$. 

A sophisticated algorithm for their computation is proposed in \cite{Paris2010} and in the Appendix at the end of this paper we show the main step for the application to the Prabhakar function; thanks to this algorithm we are able to numerically evaluate any number of $c_k$'s. The first few coefficients are however explicitly listed here

\[
	\begin{aligned}
	c_{0} &= 1 \\
	c_{1} &= \frac{(\gamma-1)}{2} \left( \alpha \gamma + \gamma - 2 \beta\right) \\
	c_{2} &= \frac{(\gamma-1)(\gamma-2)}{24} \left( 3(\alpha+1)^2\gamma^2 - (\alpha+1)(\alpha+12\beta+5) \gamma + 12\beta(1+\beta) \right) \\
	\end{aligned}
\]

\begin{remark}
When $\gamma=1$ it is $c_{0}=1$ and $c_{k}=0$, $k=1,2,\dots$; in this case it is indeed possible to verify that for the functions introduced in equation (\ref{eq:UR_Definition}) in the Appendix it is $R(s)=\Upsilon(s)\equiv1$ and hence in equation (\ref{eq:Coeff_ML2}) we have $R_0=\Upsilon_0=1$ and $R_k=\Upsilon_k=0$ for any $k=1,2,\dots$. Thus, the expansion given in Theorem \ref{thm:Expansion1} corresponds to the well-known expansion (e.g., see \cite[Theorem 4.3]{GorenfloKilbasMainardiRogosin2014}) for the standard two parameter ML function.
\end{remark}

\subsection{Expansion on the negative semi-axis}

In several applications it is of special interest to investigate the asymptotic behavior of the Prabhakar function along the real and negative semi-axis. For example, in a recent paper, Xu \cite{Xu} has studied particle deposition in porous media, comparing different models involving Riemann-Liouville, Hadamard and Prabhakar time-fractional derivatives. In the last case, the analysis of the asymptotic behavior of the Prabhakar function plays a relevant role. To this purpose, on the basis of the results from the previous subsection we can provide the following result.

\begin{theorem}\label{thm:ExpNegativeSemiAxis}
Let $\alpha >0$ and $t > 0$. Then 
\[	
	E_{\alpha,\beta}^{\gamma}(-t) \sim 
		\left\{\begin{array}{lc}
		H(t) & 0<\alpha<2 \\
		C_{0}(t) + H(t)  & \alpha = 2 \\ 
		\displaystyle\sum_{r=0}^{P-1} C_r(t)  \quad \quad & \alpha > 2\\ 
		\end{array} \right.
	, \quad \text{as } t\to \infty,
\]
with $P = {\textstyle \left\lfloor \frac{1}{2} \left( \frac{\alpha}{2}+1\right) \right\rfloor }$ and 
\[
	C_{r}(t) =  \frac{2}{\alpha^{\gamma}\Gamma(\gamma)} \exp\Bigl(t^{\frac{1}{\alpha}} \cos {\textstyle \frac{(2r+1)\pi}{\alpha}} \Bigr) \sum_{k=0}^{\infty} c_k t^{\frac{\gamma-\beta-k}{\alpha}} \cos \Bigl( {\textstyle\frac{(2r+1)\pi(\gamma-\beta-k)}{\alpha}}+ t^{\frac{1}{\alpha}} \sin {\textstyle \frac{(2r+1)\pi}{\alpha}} \Bigr) .
\]
\end{theorem}

\begin{pf}
By applying Theorems \ref{thm:Expansion1}, \ref{thm:Expansion2} and \ref{thm:Expansion3}, and after using the equivalence $-t=\eu^{\pi \iu}t$, as $t \to \infty$ we can write 
\[
	E_{\alpha,\beta}^{\gamma}(-t) \sim \left\{ \begin{array}{ll}
		H(t) & 0<\alpha<2 \\
		F(\eu^{\pi \iu} t) + F(\eu^{-\pi\iu} t) + H(t)  & \alpha = 2 \\
		\displaystyle\sum_{r=-P}^{P} F(\eu^{\pi\iu (2r+1)} t) , \quad {\textstyle P = \left\lfloor \frac{1}{2} \left( \frac{\alpha}{2}+1\right) \right\rfloor }		\quad	& 2 < \alpha ,  \\
	\end{array} \right.
\]
with $\left\lfloor x \right\rfloor$ the largest integer smaller than $x$. It is simple to observe that
\[
	\exp \left( \bigl(\eu^{ \iu r \pi} t\bigr)^{\frac{1}{\alpha}}\right) 
	= \exp\Bigl( t^{\frac{1}{\alpha}} \cos {\textstyle \frac{r\pi}{\alpha}} \Bigr) 
		\left[ \cos \Bigl( t^{\frac{1}{\alpha}} \sin {\textstyle \frac{r\pi}{\alpha}} \Bigr) + \iu \sin \Bigl( t^{\frac{1}{\alpha}} \sin{\textstyle \frac{r\pi}{\alpha}} \Bigr) \right]
\]
and
\[
	\bigl( \eu^{ \iu r \pi } t\bigr)^{\frac{\gamma-\beta-k}{\alpha}}
	 = t^{\frac{\gamma-\beta-k}{\alpha}} \left[ \cos {\textstyle \frac{r\pi (\gamma-\beta-k)}{\alpha}} + \iu \sin {\textstyle \frac{r\pi(\gamma-\beta-k)}{\alpha}} \right]
\]
and hence a standard manipulation leads to 
\[
	\begin{aligned}
		G_r(t) 
		:= &F(\eu^{\pi r \iu} t) + F(\eu^{- \pi r \iu} t) = \\
		= &\frac{2}{\alpha^{\gamma}\Gamma(\gamma)} \exp\Bigl(t^{\frac{1}{\alpha}} \cos {\textstyle \frac{r\pi}{\alpha}} \Bigr) \sum_{k=0}^{\infty} c_k t^{\frac{\gamma-\beta-k}{\alpha}} \cos \Bigl( {\textstyle\frac{r\pi(\gamma-\beta-k)}{\alpha}}+ t^{\frac{1}{\alpha}} \sin {\textstyle \frac{r\pi}{\alpha}} \Bigr) \\
	\end{aligned}
\]
and, for $r=0,1,\dots,P-1$, we can put $C_{r}(t)=G_{2r+1}(t)$.

Moreover, let us consider
\[
\begin{aligned}
	F(\eu^{\pi (2P+1) \iu} t) =&  \frac{1}{\alpha^{\gamma}\Gamma(\gamma)} \exp\Bigl(t^{\frac{1}{\alpha}} \cos {\textstyle \frac{(2P+1)\pi}{\alpha}} \Bigr) \sum_{k=0}^{\infty} c_k t^{\frac{\gamma-\beta-k}{\alpha}} \cdot \\
	& \cdot \left[ \cos \Bigl( {\textstyle\frac{(2P+1)\pi(\gamma-\beta-k)}{\alpha}}+ t^{\frac{1}{\alpha}} \sin {\textstyle \frac{(2P+1)\pi}{\alpha}} \Bigr) + \iu \sin \Bigl( {\textstyle\frac{(2P+1)\pi(\gamma-\beta-k)}{\alpha}}+ t^{\frac{1}{\alpha}} \sin {\textstyle \frac{(2P+1)\pi}{\alpha}} \Bigr) \right] \\
	\end{aligned}	
\]
and clearly $F(\eu^{\pi (2P+1) \iu} t) = \frac{1}{2}\bigl(C_{P}(t) + \iu S_{P}(t)\bigr)$ with
\[
	S_{r}(t) = \frac{2}{\alpha^{\gamma}\Gamma(\gamma)} \exp\Bigl(t^{\frac{1}{\alpha}} \cos {\textstyle \frac{(2r+1)\pi}{\alpha}} \Bigr) \sum_{k=0}^{\infty} c_k t^{\frac{\gamma-\beta-k}{\alpha}} \sin \Bigl( {\textstyle\frac{(2r+1)\pi(\gamma-\beta-k)}{\alpha}}+ t^{\frac{1}{\alpha}} \sin {\textstyle \frac{(2r+1)\pi}{\alpha}} \Bigr) .
\] 
 
When $\alpha > 2$ the value $P$ is such that $\alpha < 2 (2P+1)$ and hence $\cos {\textstyle \frac{(2P+1)\pi}{\alpha}} < 0$; thus $C_{P}(t)\to0$ and $S_{P}(t)\to 0$ as $t\to \infty$ and therefore they can be neglected in the expansion of $E_{\alpha,\beta}^{\gamma}(-t)$ as compared to the other values of $C_r(t)$ in which exponential functions with positive arguments appear. \qed
\end{pf}

For the sake of clearness we plot in Figure \ref{fig:Fig_Alpha_P} the value $P$ of the number of terms $C_r(t)$ in dependence of $\alpha$.

\begin{figure}[ht]
\centering
\includegraphics[width=0.50\textwidth]{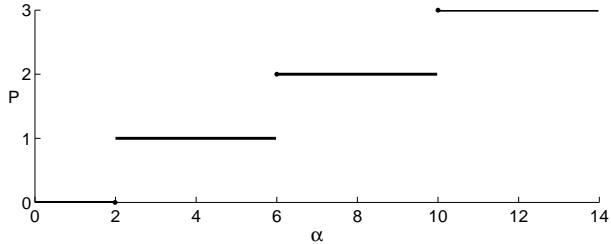}
\caption{Number of terms $C_r(t)$ in the expansion of $E_{\alpha,\beta}^{\gamma}(-t)$ in dependence of $\alpha$.}
\label{fig:Fig_Alpha_P}
\end{figure}	

Observe that when $\alpha=2$ the exponential in $C_0(t)$ gives a constant contribution since its argument is zero; thus, only terms decaying in an algebraic way appear for $0<\alpha\le 2$. For $\alpha >2$ we observe that $\cos {\textstyle \frac{(2r+1)\pi}{\alpha}}$ decreases as $r$ increases and hence just the first few functions $G_r(t)$ are sufficient to describe in an accurate way the behavior of $E_{\alpha,\beta}^{\gamma}(-t)$.

\section{Nonlinear heat conduction equations with memory involving Prabhakar derivatives.}\label{S:Heat}
        
Prabhakar-type fractional derivatives have gained interest in the recent literature not only in connection to anomalous dielectric relaxation models. For instance, Bulavatsky in \cite{Bulavatsky2017} considered mathematical models of filtration dynamics based on fractional equations involving the Hilfer-Prabhakar derivatives introduced in \cite{GarraGorenfloPolitoTomovski2014,PolitoTomovski2016}. In all these cases the authors considered essentially linear models of diffusion or relaxation with memory, treated by means of fractional differential equations involving Prabhakar-type derivatives. 

In this paper we extend the analysis by considering for the first time the application of Prabhakar derivatives in the theory of nonlinear heat conduction equation with memory.

It is well-known that the heat equation arises from the energy balance equation (in the one dimensional case)
\begin{equation}\label{eq:HeatEq}
\rho c \frac{\partial T}{\partial t}= -\frac{\partial q}{\partial x}
\end{equation}
together with the Fourier law
\begin{equation}\label{eq:FourierLaw}
q(x,t)= -k_T\frac{\partial T}{\partial x},
\end{equation}
where $\rho$ and $c$ are the medium mass density and the medium specific heat respectively, while $k_T$ is the thermal conductivity.

Starting from the general theory of heat conduction with finite wave speeds of Gurtin and Pipkin \cite{GurtinPipkin1968}, many applications of Caputo time-fractional derivatives to problems of heat propagation with memory effects have been studied in the literature (see for example \cite{Povstenko2015} for a review and \cite{FabrizioGiorgiMorro2017} for a useful physical discussion).
On the other hand, the temperature dependence of the thermal conductivity is in many cases not negligible, leading to a nonlinear formulation of the classical heat equation (see for example \cite{Straughan2011}).

In the context of the theory of heat propagation with memory, we here consider a new formulation of the energy balance (\ref{eq:HeatEq}), by replacing the first time derivative with the Caputo-type Prabhakar derivative (\ref{CP}), with $\alpha \gamma \in (0,1)$, and considering the following power law temperature dependence of the thermal coefficient in (\ref{eq:FourierLaw}) \cite{Straughan2011}
\begin{equation}
k_T\sim k_0 T^\xi, \quad \xi >0.
\end{equation}

Therefore we obtain the following nonlinear fractional heat conduction equation
\begin{equation}\label{HE}
{}^{\text{\tiny{C}}}{}{\bigl({}_{0}D^{\alpha}_{t} + \lambda \bigr)^{\gamma}} T(x,t)= k_0 \frac{\partial}{\partial x} T^\xi \frac{\partial T}{\partial x}-\beta T, \quad \beta, \xi >0
\end{equation}
where we have also included a linear term in order to represent the linear heat loss.

In few cases it is possible to find exact analytical solutions to nonlinear fractional partial differential equations such as \eqref{HE}. In the recent literature, useful tools in this direction are provided by the Lie group analysis and the invariant subspace method (see for example \cite{SahadevanPrakash2016} and the references therein).

The aim of this section is therefore dual: to consider a new nonlinear model of heat propagation with memory involving Prabhakar functions and to provide some exact results by using a generalized separating variable method (see for example the Encyclopedic book of Polyanin and Zaitsev \cite{PolyaninZaitsev2004}).

\begin{proposition}
Let $C$ an arbitrary positive constant depending on the initial condition and $\alpha \gamma \in (0,1)$, then the nonlinear fractional heat conduction equation \eqref{HE} admits as a particular solution the following one
\begin{equation}\label{sol}
T(x,t)= (x+C)^{\frac{1}{1+\xi}}\sum_{k=0}^\infty (-\beta)^k t^{\alpha \gamma k}E^{\gamma k}_{\alpha, 1+\alpha\gamma k}(-\lambda t^{\alpha}).
\end{equation}
\end{proposition}

\begin{proof}
We start our proof from the following \textit{ansatz} on the separating variable solution for \eqref{HE}
\begin{equation}\label{ans}
T(x,t)= f(t)(x+C)^{\frac{1}{1+\xi}}.
\end{equation}

Indeed, it is simple to prove that the solution \eqref{ans} is such that
\begin{equation}
k_0 \frac{\partial}{\partial x} T^\xi \frac{\partial T}{\partial x}= 0
\end{equation}
and therefore, by substituting \eqref{ans} in \eqref{HE}, the problem is reduced to the search of the solution of the following linear fractional differential equation
\begin{equation}\label{e2}
{}^{\text{\tiny{C}}}{}{\bigl({}_{0}D^{\alpha}_{t} + \lambda \bigr)^{\gamma}} f(t)=-\beta f(t).
\end{equation}

Assuming that $f(0)= 1$ (and therefore the corresponding initial condition of our problem is $T(x,0)= (x+C)^{\frac{1}{1+\xi}}$), we apply the LT method to solve \eqref{e2}. Recalling that (see e.g. \cite{GarrappaMainardiMaione2016})
\begin{equation}
\mathcal{L}({}^{\text{\tiny{C}}}{}{\bigl({}_{0}D^{\alpha}_{t} + \lambda \bigr)^{\gamma}} f(t);s) = (s^\alpha+
\lambda)^\gamma \tilde{f}(s)-s^{-1}(s^\alpha+
\lambda)^\gamma f(0^+),
\end{equation}
we have that 
\[
\tilde{f}(s) = s^{-1}\left(1+\frac{\beta}{(s^\alpha+\lambda)^\gamma}\right)^{-1} = \sum_{k=0}^{\infty}(-\beta)^k s^{-1-\alpha \gamma k}
(1+\frac{\lambda}{s^\alpha})^{-\gamma k}.
\]

Recalling that (see for example \cite{KilbasSaigoSaxena2004} and \cite{GorenfloKilbasMainardiRogosin2014})
\begin{equation}
\mathcal{L}(t^{\mu-1}E^{-\gamma}_{\rho,\mu}(\omega t^\rho);s) = s^{-\mu}(1-\omega s^{-\rho})^{\gamma},
\end{equation}
we finally obtain 
\begin{equation}\label{eq:SolutionTime}
f(t)= \sum_{k=0}^\infty (-\beta)^k t^{\alpha \gamma k}E^{\gamma k}_{\alpha, 1+\alpha\gamma k}(-\lambda t^{\alpha}),
\end{equation}
as claimed.
\end{proof}

We refer to \cite{GarraGorenfloPolitoTomovski2014} for the discussion of the convergence of the series appearing in \eqref{sol}.

On the same line of the previous problem, we can consider also another interesting case, in which the thermal coefficient has an exponential dependence by temperature, i.e. 
\begin{equation}\label{HE1a}
k_T\sim k_0 e^{\nu T/f(t)}, \quad \nu>0 ,
\end{equation}
with $f(t)$ the function defined in (\ref{eq:SolutionTime}), leading to the following nonlinear fractional heat conduction equation
\begin{equation}\label{HE1}
{}^{\text{\tiny{C}}}{}{\bigl({}_{0}D^{\alpha}_{t} + \lambda \bigr)^{\gamma}} T(x,t)= k_0 \frac{\partial}{\partial x} \bigg(e^{\nu T/f(t)} \frac{\partial T}{\partial x}\bigg)-\beta T. 
\end{equation}

In this case we have the following result (we leave the proof to the reader, since it is essentially based on the same ideas of the previous one).

\begin{proposition}
Let $C$ an arbitrary positive constant depending on the initial condition, $\beta, \nu>0$ and $x>C$, then the nonlinear fractional heat conduction equation \eqref{HE1} admits as a particular solution the following one
\begin{equation}\label{sol1}
T(x,t)= \frac{\ln(x+C)}{\nu}\sum_{k=0}^\infty (-\beta)^k t^{\alpha \gamma k}E^{\gamma k}_{\alpha, 1+\alpha\gamma k}(-\lambda t^{\alpha}) .
\end{equation}
\end{proposition}

The function $f(t)$ introduced by equation (\ref{eq:SolutionTime}) is actually the eigenfunction of the Prabhakar derivative ${}^{\text{\tiny{C}}}{}{\bigl({}_{0}D^{\alpha}_{t} + \lambda \bigr)^{\gamma}}$ and its analysis is of importance to fully understand the behavior of the solutions of differential equations with Prabhakar derivatives. 

 We decided to consider these particular solutions with the aim of suggesting the possible utility of the asymptotic estimates here studied in the context of nonlinear physics. 
We can infer a posteriori the physical meaning of these solutions, observing that they correspond to the evolution of a temperature field with an initially logarithmic or power law distribution.

 We finally observe that symmetry methods allow to find other classes of solutions, e.g. of similarity type (see for example \cite{SahadevanPrakash2016} and the references therein). However, as far as we know, there are not any studies about other exact solutions for the nonlinear fractional diffusion equations here considered.

 In the case of most interest, i.e. when $0<\alpha<2$, from Theorem \ref{thm:ExpNegativeSemiAxis} we know that 
\[
	E^{\gamma}_{\alpha, \beta}(-\lambda t^{\alpha}) \sim \sum_{j=0}^{\infty} \frac{(-1)^{j} \Gamma(j+\gamma)}{\Gamma(\gamma) j! \Gamma(\beta-\alpha(j+\gamma))} t^{-\alpha \gamma-j} \lambda^{-\gamma-j}
	, \quad t \to \infty ,
\]
and, hence, some simple manipulations allow to express the asymptotic behavior of $f(t)$ as
\[
	f(t) 
	%= \sum_{k=0}^\infty (-\beta)^k t^{\alpha \gamma k}E^{\gamma k}_{\alpha, 1+\alpha\gamma k}(-\lambda t^{\alpha})
	\sim \sum_{j=0}^{\infty} \frac{ \bigl( t^{\alpha}\lambda\bigr)^{-j}}{ \Gamma(1-\alpha j)} \varphi_j(\beta,\lambda,\gamma)
	, \quad
	t \to \infty ,
\]
where the coefficients $\varphi_j(\beta,\lambda,\gamma)$ are clearly given by
\[
	\varphi_j(\beta,\lambda,\gamma) = \frac{(-1)^j}{j!}\sum_{k=0}^\infty \frac{\Gamma(j+\gamma k)}{\Gamma(\gamma k)} \left( - \frac{\beta}{\lambda^{\gamma}} \right)^k .
\]

An approximation of the first coefficients $\varphi_j(\beta,\lambda,\gamma)$ can be numerically evaluated, thus allowing to draw the asymptotic behavior of $f(t)$ for large $t$; the corresponding plots are presented in the left plots of Figures \ref{fig:Fig_Calore_al}, \ref{fig:Fig_Calore_ga}, \ref{fig:Fig_Calore_la} and \ref{fig:Fig_Calore_be} for $\alpha$, $\gamma$, $\lambda$ and $\beta$ varying respectively. 

We must note that all the solutions converge to $\varphi_0(\beta,\lambda,\gamma)$.  To illustrate the same behavior in a logarithmic scale (see the corresponding plots in the right parts of each Figure) we have to first observe that
\[
	f(t) 	\sim \varphi_0(\beta,\lambda,\gamma) + t^{-\alpha} \frac{\varphi_1(\beta,\lambda,\gamma)}{\lambda \Gamma(1-\alpha)}
	, \quad
	t \to \infty ,
\]
where it is immediate to evaluate 
\[
	\varphi_0(\beta,\lambda,\gamma) = \frac{\lambda^{\gamma}}{\lambda^{\gamma} + \beta}
	, \quad
	\varphi_1(\beta,\lambda,\gamma) = \frac{ \beta \gamma \lambda^{\gamma}}{(\lambda^{\gamma} + \beta)^{2}} ;
\]
it is hence possible to highlight the power-low decay by drawing the function
\[
	\tilde{f}(t) = f(t) - \varphi_0(\beta,\lambda,\gamma) 
\]	
in the logarithmic scale, as shown in the right plots of Figures  \ref{fig:Fig_Calore_al}, \ref{fig:Fig_Calore_ga}, \ref{fig:Fig_Calore_la} and \ref{fig:Fig_Calore_be} (in each figure the parameters $\alpha$, $\gamma$, $\lambda$ and $\beta$ are chosen just with the aim of presenting clear and distinguishable plots). 

\begin{figure}[ht]
\centering
\begin{tabular}{cc}
\includegraphics[width=0.50\textwidth]{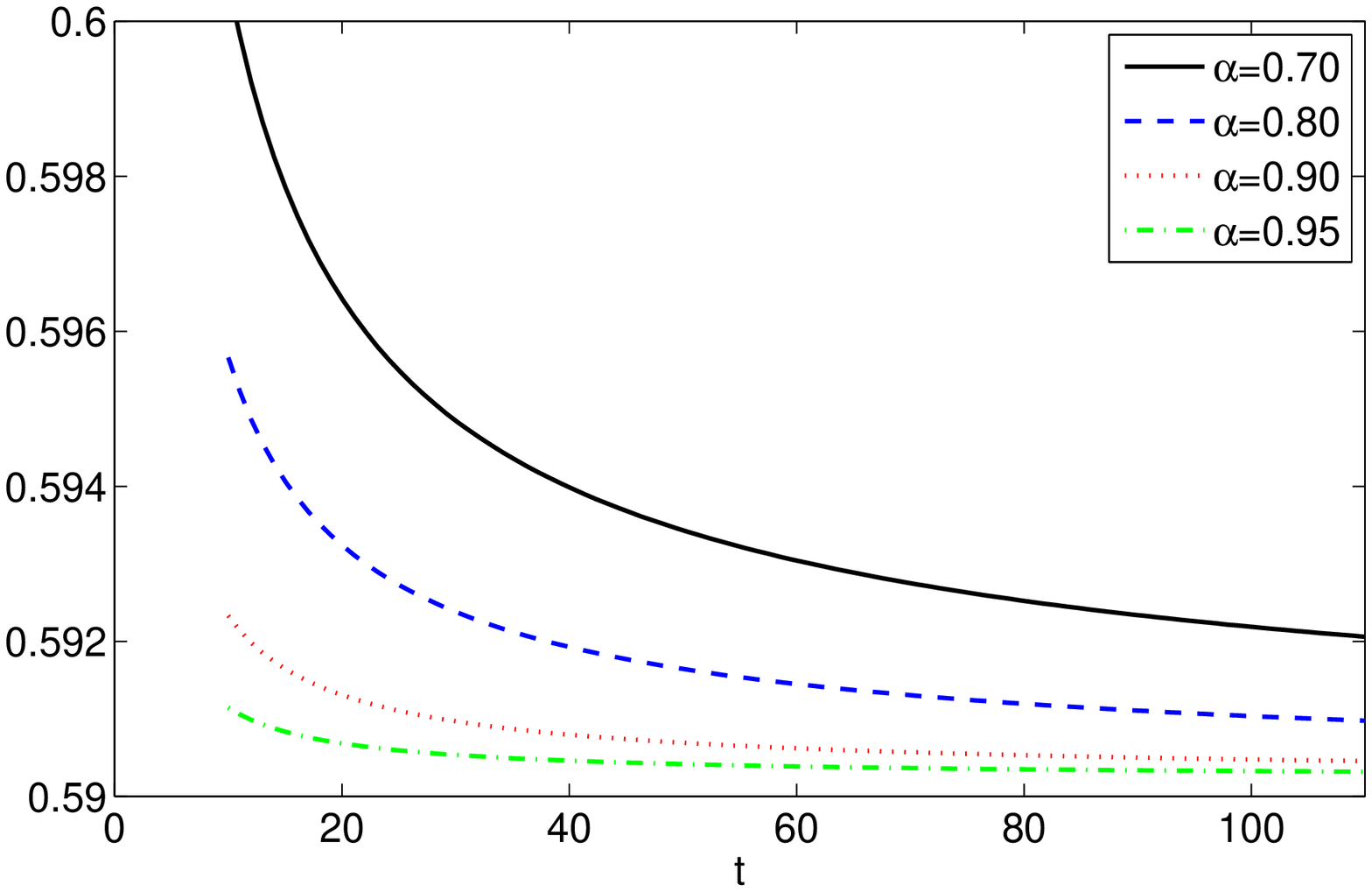} &
\includegraphics[width=0.50\textwidth]{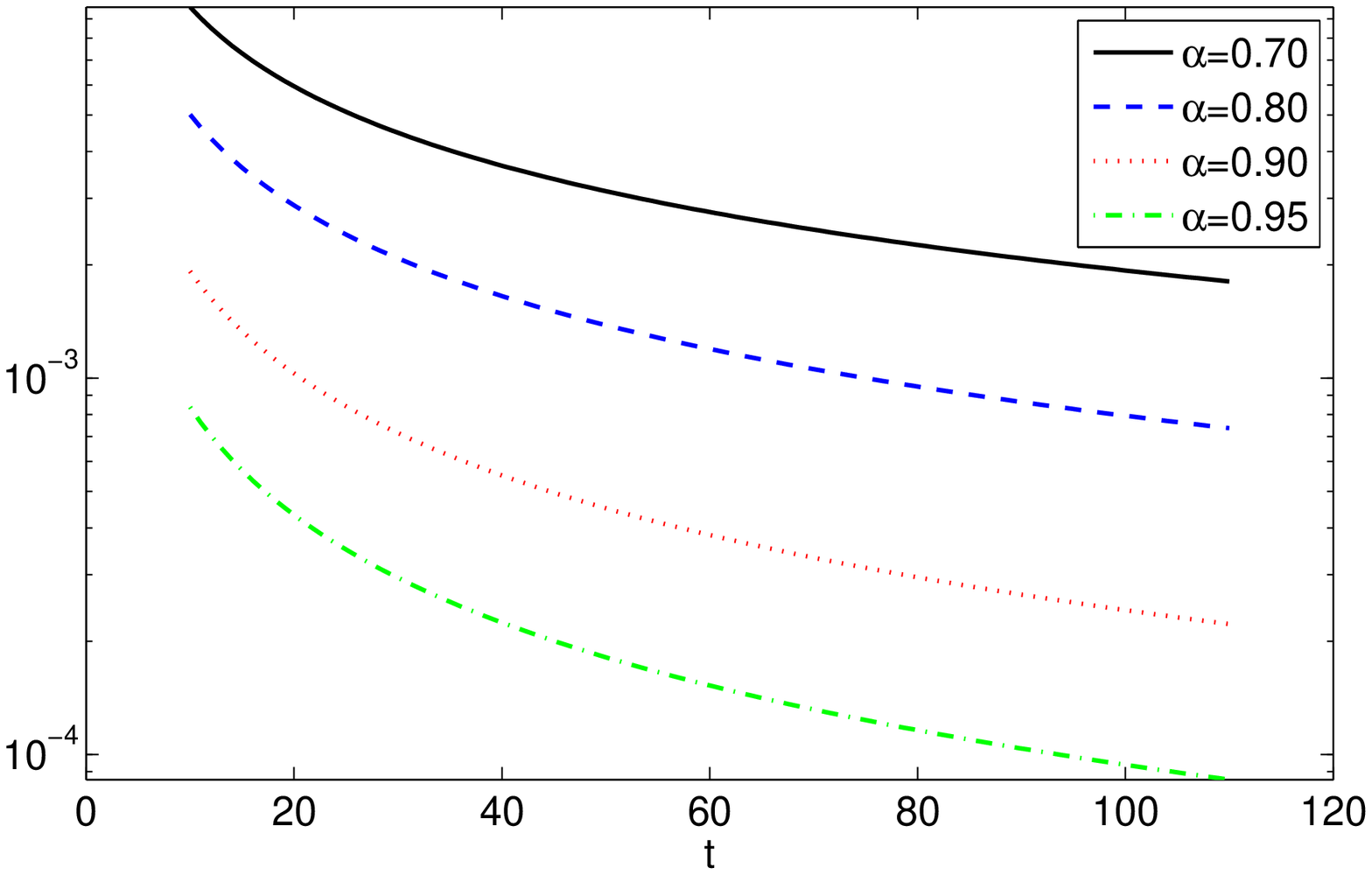}
\end{tabular}
\caption{Behavior of $f(t)$ (left plot) and of $\tilde{f}(t)$ (right plot) as $\alpha$ varies (here $\gamma=0.9$, $\lambda=1.5$ and $\beta=1.0$).}
\label{fig:Fig_Calore_al}
\end{figure}

\begin{figure}[ht]
\centering
\begin{tabular}{cc}
\includegraphics[width=0.50\textwidth]{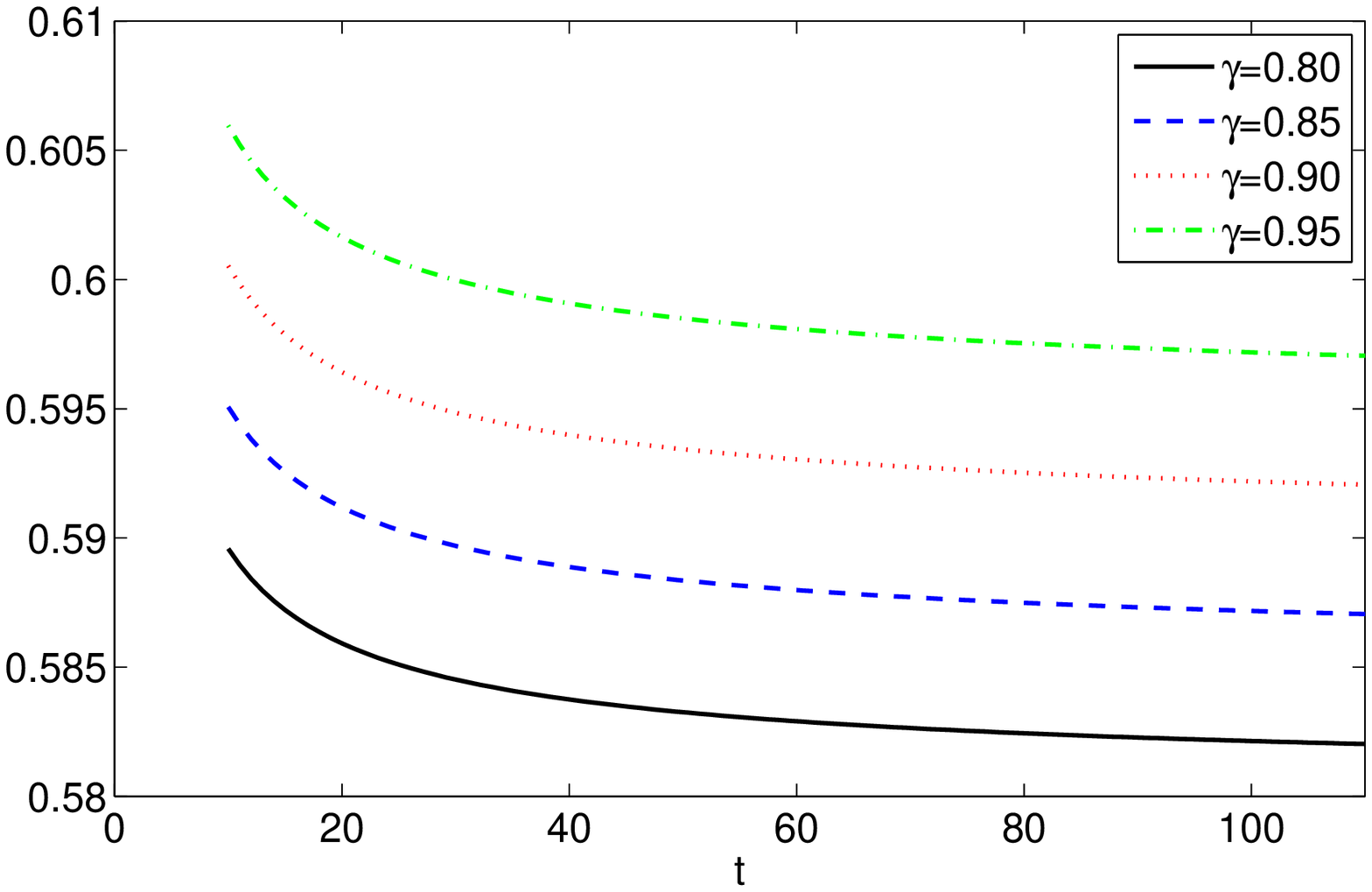} &
\includegraphics[width=0.50\textwidth]{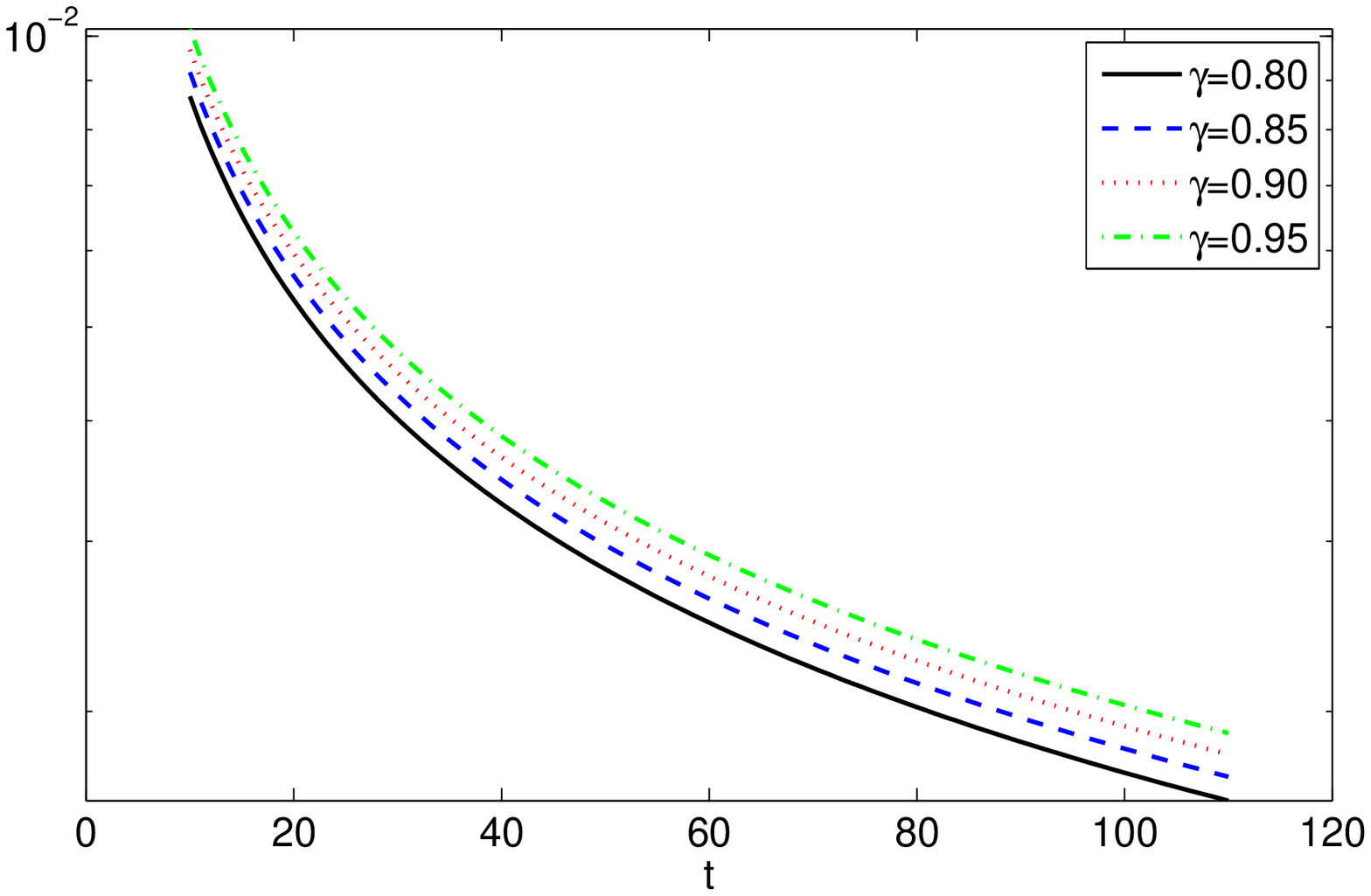}
\end{tabular}
\caption{Behavior of $f(t)$ (left plot) and of $\tilde{f}(t)$ (right plot) as $\gamma$ varies (here $\alpha=0.7$, $\lambda=1.5$ and $\beta=1.0$).}
\label{fig:Fig_Calore_ga}
\end{figure}

\begin{figure}[ht]
\centering
\begin{tabular}{cc}
\includegraphics[width=0.50\textwidth]{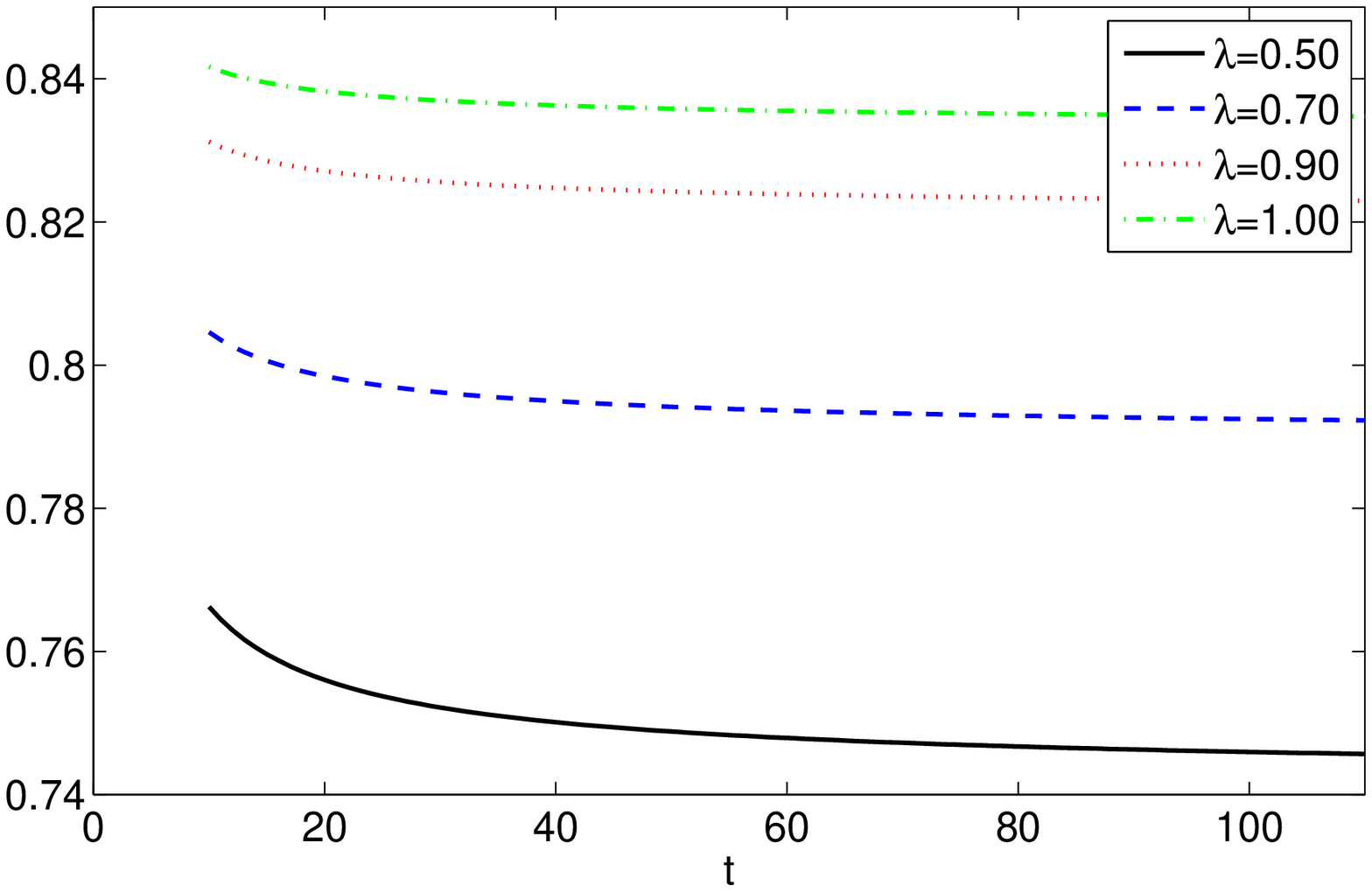} &
\includegraphics[width=0.50\textwidth]{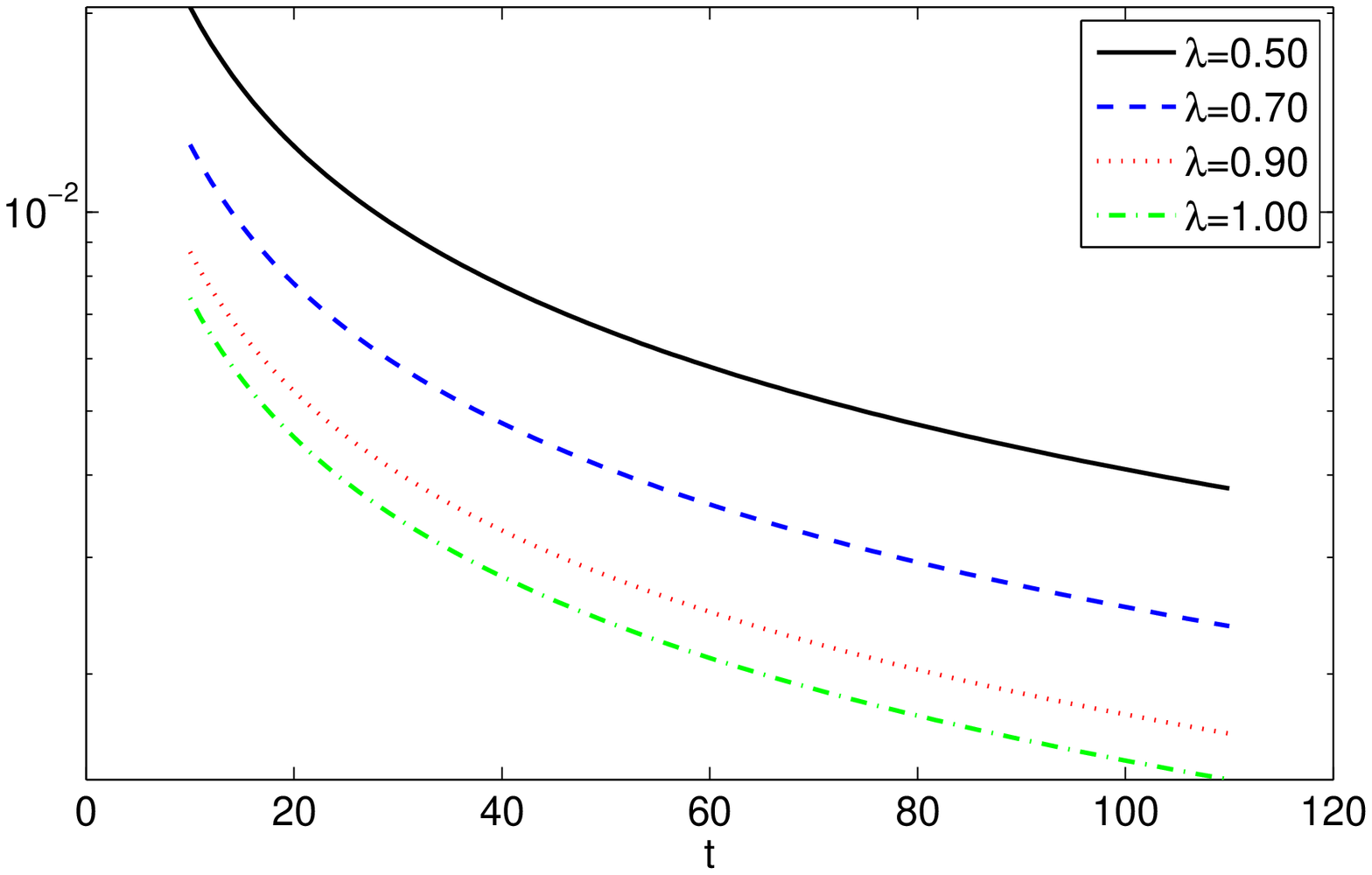}
\end{tabular}
\caption{Behavior of $f(t)$ (left plot) and of $\tilde{f}(t)$ (right plot) as $\lambda$ varies (here $\alpha=0.7$, $\gamma=0.8$ and $\beta=0.2$).}
\label{fig:Fig_Calore_la}
\end{figure}

\begin{figure}[ht]
\centering
\begin{tabular}{cc}
\includegraphics[width=0.50\textwidth]{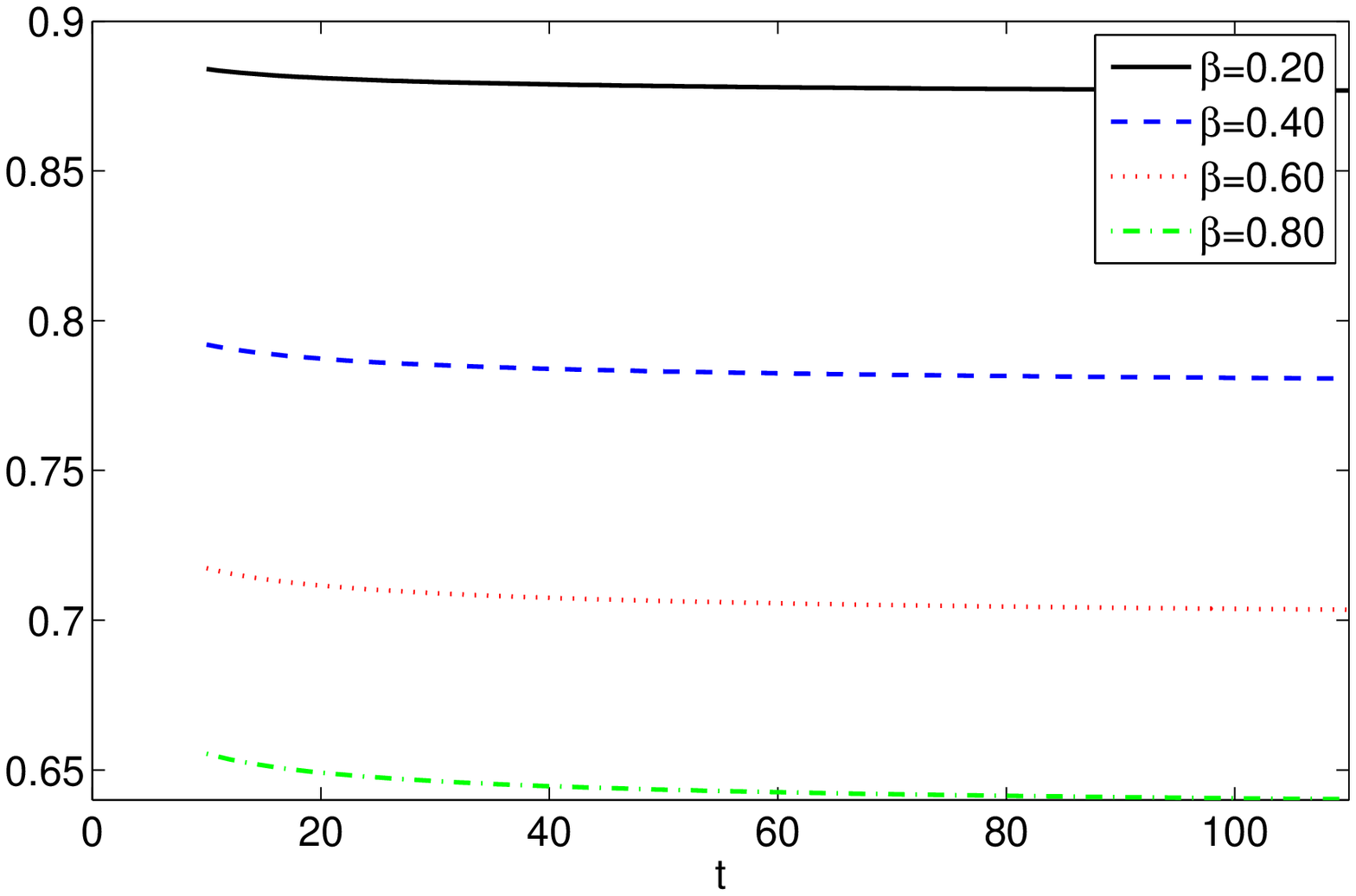} &
\includegraphics[width=0.50\textwidth]{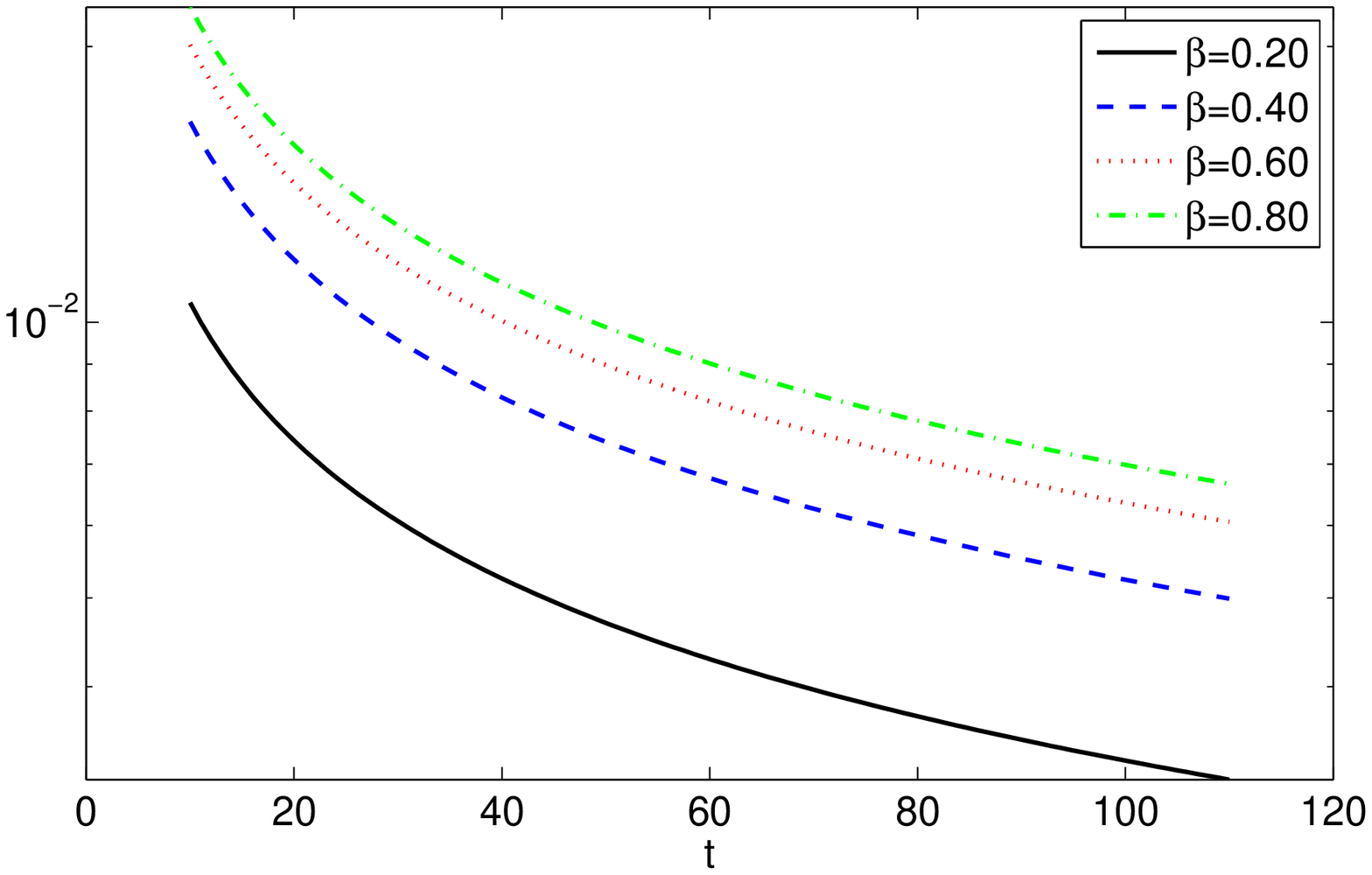}
\end{tabular}
\caption{Behavior of $f(t)$ (left plot) and of $\tilde{f}(t)$ (right plot) as $\beta$ varies (here $\alpha=0.5$, $\gamma=0.8$ and $\lambda=1.5$).}
\label{fig:Fig_Calore_be}
\end{figure}

\section{Concluding remarks}\label{S:Concluding}

In this paper we have studied the Prabhakar function, with a particular attention to the asymptotic behavior in the whole complex plane and on the negative semi-axis. We have also described fractional integral and fractional derivatives of Prabhakar type and studied the solution of some heat conduction equations with Prabhakar derivatives.

Fractional derivatives of Prabhakar type are stimulating the interest of researchers in several areas; the presence of two real powers allows indeed a better fitting of experimental data and the development of more reliable models. We think therefore that contributions devoted to study the behavior of the Prabhakar function, and of the solution of the equations based on Prabhakar operators, are surely useful for a fully understating of the dynamic of these models.

A further investigation should concern in the future on methods for the numerical evaluation the Prabhakar function, a task which has been performed only in a partial way in \cite{Garrappa2015_SIAM,StanislavskyWeron2012} and on numerical schemes for the solution of differential equations with operators of Prabhakar type.

\appendix
\section{Evaluation of coefficients in the asymptotic expansion}\label{S:AppendixEvaluation}

The main difficulty when using the asymptotic expansions presented in Theorems \ref{thm:Expansion1}, \ref{thm:Expansion2} and \ref{thm:Expansion3} is the evaluation of the coefficients $c_{j}$ in (\ref{eq:InverseFactorialExpansion}). Not only they depend on the three parameters $\alpha$, $\beta$ and $\gamma$ but in addition their evaluation is not a trivial task. It is indeed possible to provide an analytic representation in a reasonable form only for the first instances of $c_{j}$ whilst for the subsequent coefficients it is necessary to proceed by numerical evaluation.

We review, therefore, the procedure devised in \cite{Paris2010} with the aim of illustrating the main steps of an algorithm for the automatic generation of any number of coefficients $c_j$ in the asymptotic expansions of the Prabhakar function.

%We review here the procedure devised in \cite{Paris2010} to determine these coefficients with the aim of showing an explicit representation of the first few coefficients in \ref{eq:InverseFactorialExpansion}) and, at the same time, providing a straightforward procedure to automatically generate any number of coefficients.

The starting point is based on the observation that 
\begin{equation}\label{eq:MainStep}
	\frac{\Gamma(\gamma+s)\Gamma(\alpha s + \psi)}{\Gamma(s+1)\Gamma(\alpha s +\beta)} = \alpha^{1-\gamma} R(s) \Upsilon(s) ,
\end{equation}
where 
\begin{equation}\label{eq:UR_Definition}
	R(s) = \frac{e(s; \gamma) e(\alpha s; \psi)}{e(s;1)e(\alpha s; \beta)} ,
	\quad
	\Upsilon(s) = \frac
		{\Gamma^{\star}(\gamma+s)\Gamma^{\star}(\alpha s + \psi)}
		{\Gamma^{\star}(s+1)\Gamma^{\star}(\alpha s + \beta) }
	, 	
\end{equation}
with 
\[
	e(a; b) = \exp \left( \left(a + b - \frac{1}{2} \right) \log \left(1+\frac{b}{a} \right) - b \right) 
\]
and $\Gamma^{\star}(z)$ is the scaled gamma function defined by
\begin{equation}\label{eq:ScaledGamma}
	\Gamma^{\star}(z) = \left( 2 \pi \right)^{-\frac{1}{2}} \eu^z z^{\frac{1}{2}-z} \Gamma(z) .
\end{equation}

Hence, once asymptotic expansions for $R(s)$, $\Upsilon(s)$ and the reciprocal of the rising factorial $(\alpha s + \psi)_j$ are available, namely
\begin{equation}\label{eq:UR_Expansion}
	R(s) = \sum_{k=0}^{\infty} R_k s^{-k} 
	, \quad
	\Upsilon(s) = \sum_{k=0}^{\infty} \Upsilon_j s^{-k} 
	, \quad
	\frac{1}{(\alpha s + \psi)_j} = \sum_{k=0}^{\infty} D_{j,k} s^{-k} ,
\end{equation}
and since it is immediate to see that
\[
	\sum_{j=0}^{\infty}  \frac{c_j}{(\alpha s + \psi)_{j}} 
	= \sum_{j=0}^{\infty} c_{j} \sum_{k=0}^{\infty}  D_{j,k} s^{-k} 
	= \sum_{k=0}^{\infty} \left( \sum_{j=0}^{k} c_{j} D_{j,k} \right) s^{-k} ,
\]
where the last equality holds since, as we will see later, it is $D_{j,k} = 0$ for $j > k$, by matching terms by terms these expansions in the equivalence (\ref{eq:MainStep}) we obtain for any $k=0,1,\dots$
\[
	\sum_{j=0}^{k} c_{j} D_{j,k} = \sum_{j=0}^{k} R_{k-j} \Upsilon_{j}
\]
from which it is immediate to obtain the coefficients $c_j$ by means of the recursive relationship
\begin{equation}\label{eq:Coeff_ML2}
	c_{0} = \frac{R_{0} \Upsilon_{0}}{D_{0,0}} 
	, \quad
	c_{k} = \frac{1}{D_{k,k}} \left( \sum_{j=0}^{k} R_{k-j} \Upsilon_{j} - \sum_{j=0}^{k-1} c_{j} D_{j,k} \right) .
\end{equation}

The problem therefore becomes to find the three expansions in (\ref{eq:UR_Expansion}) which will be addressed in the next three subsections.

\subsection{Expansion for the reciprocal of the rising factorial} 

The expansion for the reciprocal of the rising factorial $(\alpha s + \psi)_j$ is easily obtained by the following result on the expansion of the ratio of two gamma functions \cite{TricomiErdelyi1951} 
\[
	\frac{1}{(\alpha s + \psi)_j} 
	= \frac{\Gamma(\alpha s + \psi)}{\Gamma(\alpha s + \psi+j)} 
	= \sum_{k=0}^{\infty} \frac{1}{\alpha^{k+j}} F_{j,k} s^{-k-j} ,
\]
where, the coefficients $F_{j,k}$ can be recursively evaluated as
\[
	F_{j,0} = 1 
	, \quad
	F_{j,k} = \frac{1}{k} \sum_{\ell=0}^{k-1} \left[ \binom{-j-\ell}{k-\ell+1} + (-1)^{k-\ell} j (\psi+j)^{k-\ell} \right] F_{j,\ell} .
\]

We note that since in the binomial coefficients the first argument is negative, a generalized definition must be used according to \cite{Kronenburg2011} \[
	\binom{n}{k} = \left\{\begin{array}{ll}
		\displaystyle\frac{n!}{k!(n-k)!} & \, 0 \le k \le n \\
		%0 & 0 \le n < k \\
		\displaystyle(-1)^k \binom{-n+k-1}{k}= (-1)^k\frac{(-n+k-1)!}{k!(-n-1)!} \, & n < 0 \le k \\
		\displaystyle (-1)^{n-k} \binom{-k-1}{n-k} = (-1)^{n-k}\frac{(-k-1)!}{(n-k)!(-n-1)!} \, & k \le n < 0 \\
		0 & \text{otherwise} \
	\end{array}\right.
\]
where $n$ and $k$ are integers.

Thus, the expansion in (\ref{eq:UR_Expansion}) for the reciprocal of $(\alpha s + \psi)_j$ is obtained thanks to the coefficients
\[
	D_{j,k} = \left\{ \begin{array}{ll}
		0 & 0 \le k < j \\
		\displaystyle\frac{1}{\alpha^{k}} F_{j,k-j} & k \ge j \
	\end{array} \right. 
\]
and the first few values are reported in Table \ref{tbl:ExpansionReciprocalRisedFactorial}.

\begin{table}	 
\[
	\begin{array}{c|ccccc}
	D_{j,k} & k=0 & k=1 & k=2 & k=3 & k=4 \\ \hline
	j=0 & 1 & 0 & 0 & 0 & 0 \\
	j=1 & 0 & \frac{1}{\alpha} & \frac{\psi}{\alpha^2} & \frac{2\psi}{\alpha^3} & \frac{3\psi^2-4\psi}{3\alpha^4} \\
	j=2 & 0 & 0 & \frac{1}{\alpha^2} & \frac{2\psi+1}{\alpha^3} & \frac{-\psi^2+5\psi+3}{\alpha^4} \\
	\displaystyle j=3 & 0 & 0 & 0 & \frac{1}{\alpha^3} & \frac{3(\psi+1)}{\alpha^4} \\
	j=4 & 0 & 0 & 0 & 0 & \frac{1}{\alpha^4} \
	\end{array}
\]
\caption{First coefficients $D_{j,k}$ in the expansion in (\ref{eq:UR_Expansion}) of $1/(\alpha s + \psi)_j$} \label{tbl:ExpansionReciprocalRisedFactorial}
\end{table}

\subsection{Asymptotic expansion for $R(s)$}

Let us first denote
\[
	f(s;a,b) = \left(a s + b - \frac{1}{2} \right) \log \left(1+\frac{b}{a s} \right) - b
\]
in order to represent $e(as;b)$ and its reciprocal $e(as;b)^{-1}$ as
\begin{equation}\label{eq:exp_e_expansion}
	e(as;b)^{\pm 1} = \exp(\pm f(s;a,b)) = \sum_{n=0}^{\infty} (\pm 1)^n \frac{\bigl(f(s;a,b)\bigr)^n}{n!}
\end{equation}

By the Taylor expansion for the logarithmic function, it is immediate to see that for $|s| > \left| {b}/{a} \right|$ it is 
\[
	\log \left(1+\frac{b}{a s} \right) = \sum_{k=1}^{\infty} \frac{(-1)^{k+1}}{k} \left( \frac{b}{a} \right)^k  s^{-k}
	, \quad |s| > \left| \frac{b}{a} \right| ,
\]
and hence standard algebraic manipulations leads to 
\[
	f(s;a,b) = \sum_{k=1}^{\infty} d_{k} s^{-k} 
	, \quad
	d_{k} = (-1)^{k} \left( \frac{b}{a} \right)^k \left( \frac{1}{2k} - \frac{b}{k(k+1)} \right) .
\]

To evaluate the  coefficients in the expansion of $\bigl(f(s;a,b)\bigr)^n$, and hence those of $e(as;b)$, we have to distinguish three cases. When $b\not\in\{0,1\}$ it is $d_{1}\not=0$ and hence, since 
\[
	f(s;a,b) = d_{1} s^{-1} \left( 1 + \frac{d_{2}}{d_{1}} s^{-1} + \frac{d_{3}}{d_{1}} s^{-2} + \dots \right) ,
\]
it is possible to write
\[
	(f(s;a,b))^n = d_{1}^{n} \sum_{k=0}^{\infty} f_{k}^{(n)}  s^{-n-k} ,
\]
where the coefficients $f_{k}^{(n)}$ in the expansion of $(f(s;a,b))^n$ can be evaluated, thanks to the Miller's formula for powers of unitary formal power series \cite[Theorem 1.6c]{Henrici1974}, as
\[
	f_{0}^{(n)} = 1
	, \quad 
	f_{k}^{(n)} = \sum_{\ell=1}^{k} \left( \frac{(n+1)\ell}{k} -1 \right) \frac{d_{\ell+1}}{d_{1}} f_{k-\ell}^{(n)} .
\]

A simple replacement of the above expansion in (\ref{eq:exp_e_expansion}) allows us to provide the expansion
\begin{equation}\label{eq:e_expans}
	e(a s; b)^{\pm 1} = \sum_{k=0}^{\infty} e_{k} s^{-k} 
	, 
\end{equation}
where the coefficients are given by
\[	
	e_{0} = 1 
	, \quad
	e_{k} = \sum_{j=0}^{k} (\pm  1 )^j \frac{d_{1}^{j} f_{k-j}^{(j)}}{j!} .
\]

When $\beta=1$, and hence $d_1 =0$ with $d_k\not=1$ for $k\ge2$, we can write
\[
	f(s;a,b) = d_{2} s^{-2} \left( 1 + \frac{d_{3}}{d_{2}} s^{-1} + \frac{d_{4}}{d_{2}} s^{-2} + \dots \right)
\]
and, in a very similar way as for the previous case, we obtain 
\[
	(f(s;a,b))^n = d_{2}^{n} \sum_{k=0}^{\infty} f_{k}^{(n)}  s^{-2n-k}
	, \quad
		f_{0}^{(n)} = 1, \quad f_{k}^{(n)} = \sum_{\ell=1}^{k} \left( \frac{(n+1)\ell}{k} -1 \right) \frac{d_{\ell+2}}{d_{2}} f_{k-\ell}^{(n)} 
\]
and hence the replacement of this expansion in (\ref{eq:exp_e_expansion}) leads to the coefficients for the expansion (\ref{eq:e_expans}) of $e(as;1)$ given by
\[
	e_{0} = 1 , \quad e_{1} = 0
	, \quad
	e_{k} = \sum_{j=0}^{\bar{k}} (\pm 1)^{j} \frac{d_{2}^{j} f_{k-2j}^{\left(j\right)}}{j!}
	, \quad
	\bar{k} = \left\{ \begin{array}{ll}
		{k}/{2} & \mbox{ even } k \\
		{(k-1)}/{2} & \mbox{ odd } k \
	\end{array} \right. .
\]

In the third case, when $b=0$, it is clearly $f(s;a,0) = 0$ and hence $e(as;b)=1$; thus, the corresponding coefficients in the expansion (\ref{eq:e_expans}) are $e_{0}=1$ and $e_k=0$ for $k=1,2,\dots$.

It is now a simple exercise to evaluate the coefficients of the expansion of $R(s)$ in (\ref{eq:UR_Expansion}) once the expansions of $e(s;\gamma)$, $e(\alpha s; 1-\gamma+\beta)$, $e(\alpha s; \beta)$ and $e(s;1)$ are known, which hold under reasonable assumptions, namely
\[
	|s| > \max \left\{ \gamma, \frac{1-\alpha+\beta}{\alpha}, \frac{\beta}{\alpha} , 1 \right\} .
\]

It is indeed possible to see that the product of four power series with coefficients $\{a_{k}\}$, $\{b_{k}\}$, $\{c_{k}\}$ and $\{d_{k}\}$ is the power series whose coefficients $\{R_{k}\}$ are given by
\begin{equation}\label{eq:Formula4Products}
	R_k = \sum_{j_{1}=0}^{k} \sum_{j_2=0}^{j_{1}} a_{j_2} b_{j_{1}-j_2}
	\sum_{j_3=0}^{k-j_{1}} c_{j_3} d_{k-j_{1}-j_3} 
\end{equation}

We provide here the explicit representation of the first three coefficients in the expansion of $R(s)$
\begin{eqnarray*}
	R_0 
	&\!\!\!=\!\!\!& 1 \\
	R_1 
	&\!\!\!=\!\!\!& \frac{1}{2} \left(  \frac{1}{\alpha}(1-\beta) \beta - \frac{1}{\alpha}(\gamma-\beta)(1-\gamma+\beta) -  \gamma(1-\gamma)  \right) \\
	R_2 
	&\!\!\!=\!\!\!& 
	\frac{1}{24 \alpha^2} \left( P_2(1-\gamma+\beta) - P_2(\beta) + \alpha^2 P_2(\gamma) \right) \\
	& & + (\gamma-1)(2\beta-\gamma) \left ( \frac{\beta(\beta-1)}{4\alpha^2} - \frac{\gamma(\gamma-1)}{4\alpha} \right)  - \frac{1}{12} \
\end{eqnarray*}
where $P_2(x) = 3x^4 - 10x^3 + 9x^2$.

\subsection{Expansion of $\Upsilon(s)$}
 
In order to asymptotically expand $\Upsilon(s)$ it is necessary to consider the expansions, as $|s|\to \infty $ and $|\arg(s)| < \epsilon$, for the scaled gamma function (\ref{eq:ScaledGamma}) and its reciprocal which are given  \cite{ParisKaminski2001} by  
\begin{equation}\label{eq:ExpansionGammaScaled}
	\Gamma^{\star}(z) = 1 + \sum_{k=1}^{\infty}  \gamma_{k} s^{-k} 
	, \quad
	\frac{1}{\Gamma^{\star}(z)} = 1 + \sum_{k=1}^{\infty} (-1)^{k} \gamma_{k} s^{-k} 
	, \quad
\end{equation}
where $\gamma_{k}$ are the coefficients in the Stirling formula for the asymptotic expansion of the factorial of positive integers and can be recursively evaluated as \cite{Nemes2010}
\[
	\gamma_{k} = (2k+1)!! b_{2k+1} 
	, \quad
	b_{0}=b_{1}=1	
	, \quad	
	b_{k} = \frac{1}{k+1} \left( b_{k-1} - \sum_{j=2}^{k-1} j b_{j} b_{k-j+1}\right) ,
\]
with $(2k+1)!!=(2k+1)\dot(2k-1)\cdots5\cdot3\cdot1$ the double factorial (we point out some differences related to the minus signs in the equations (\ref{eq:ExpansionGammaScaled}) with respect to those reported in \cite{ParisKaminski2001,Paris2010}, maybe due to some notational differences).

It is an immediate task to prove that for any integer $k \ge 1$ it is
\[
	(as + b)^{-k} = (as)^{-k} \left( 1 + \frac{b}{as}\right)^{-k} 
	= \frac{(-1)^k}{b^{k}} \sum_{j=k}^{\infty} (-1)^{j} \binom{j-1}{j-k} \left(\frac{b}{a}\right)^{j} s^{-j} 
\]
and hence, by a simple manipulation, from (\ref{eq:ExpansionGammaScaled}) we have
\[
	\Gamma^{\star}(as+b) = 1 + \sum_{j=1}^{\infty} \hat{d}_{j} s^{-j} 
	, \quad
	\frac{1}{\Gamma^{\star}(as+b)} = 1 + \sum_{j=1}^{\infty} \tilde{d}_{j} s^{-j} ,
\]
where
\[
	\hat{d}_{j} = (-1)^j \left( \frac{b}{a} \right)^j \sum_{k=1}^{j} (-1)^{k} \binom{j-1}{j-k} \frac{\gamma_{k}}{b^{k}} 
	, \quad
	\tilde{d}_{j} = (-1)^j \left( \frac{b}{a} \right)^j \sum_{k=1}^{j} \binom{j-1}{j-k} \frac{\gamma_{k}}{b^{k}} 	.
\]

The coefficients $\Upsilon_k$ in the expansion of
\[
	\Upsilon(s) = \sum_{k=0}^{\infty} \Upsilon_k s^{k}
\]
can be therefore evaluated by applying again the formula (\ref{eq:Formula4Products}) for the evaluation of the coefficients of the product of the four power series defining $\Upsilon(s)$.

In the first two columns of Table \ref{tbl:ExpansionGammaScaled} we show the first instances for $\hat{d}_{j}(a,b)$ and $\tilde{d}_{j}(a,b)$ as functions of $a$ and $b$. By inserting these values  in the formula (\ref{eq:UR_Definition}) for $\Upsilon(s)$ we are able to explicitly evaluate the first coefficients in the expansion (\ref{eq:UR_Expansion}) of $\Upsilon(s)$ according to the formula (\ref{eq:Formula4Products}), as reported in the third column of the Table \ref{tbl:ExpansionGammaScaled}.

\begin{table}
\[
	\begin{array}{c|c|c|c} \hline
	j & \hat{d}_{j}(a,b) & \tilde{d}_{j}(a,b) & \Upsilon_j \\ \hline
	1 & \frac{1}{12 a} & -\frac{1}{12 a} & 0 \\ 
	2 & \frac{1}{288 a^2} - \frac{b}{12 a^2} & \frac{1}{288 a^2} + \frac{b}{12 a^2}
		& \frac{(1-\alpha^2)(\gamma-1)}{12\alpha^2} \\
	3 & - \frac{139}{51840 a^3} - \frac{b}{144a^3} + \frac{b^2}{12a^3}  
	  & \frac{139}{51840 a^3} - \frac{b}{144a^3} - \frac{b^2}{12a^3}  
	  & \frac{(\gamma-1)}{12} \left( (\gamma+1) - \frac{2\beta - \gamma +1}{\alpha^3} \right) \\ \hline
	\end{array}
\]
\caption{First instances of the coefficients in the expansion of the scaled gamma function $\Gamma^{\star}(as+b)$ and its reciprocal $1/\Gamma^{\star}(as+b)$ and of the of the function $\Upsilon(s)$} \label{tbl:ExpansionGammaScaled}
\end{table}

 \section*{Acknowledgments}

The work of R.Garrappa has been supported by an INdAM-GNCS Project 2017.
%The work of R.Garrappa has been supported partially by an INdAM-GNCS Project 2017 and partially by the COST Action CA15225.
The authors are grateful to prof. Rajesh K. Pandey, form the Indian Institute of Technology in Varanasi (India), for his help in retrieving the biographical information on Prabhakar.

\bibliographystyle{myplain}
\bibliography{ML3_Biblio}

\end{document}